%% file: OptBribery.tex
\newcommand{\longversion}[1]{#1}
\newcommand{\shortversion}[1]{}

\documentclass[10pt,usenames,svgnames,dvipsnames]{article}

\author{Palash Dey\\Indian Institute of Technology, Kharagpur\\Email: palash.dey@cse.iitkgp.ac.in}

\input{preamble.tex}

\title{Local Distance Constrained Bribery in Voting}

\shortversion{\hypersetup{draft}}

\begin{document}

\longversion{
\maketitle
\input{abstract}
}

\shortversion{

\input{abstract}

\maketitle
}

\input{introduction}
\input{prelim}
\input{results}
\input{conclusion}

\shortversion{\bibliographystyle{ACM-Reference-Format}}  
\longversion{\bibliographystyle{alpha}}

\bibliography{references}
\end{document}

%% file: preamble.tex
\usepackage{amssymb,amsmath,amsthm}

\allowdisplaybreaks

\usepackage{xcolor}
\usepackage{color}
\shortversion{
\usepackage{hyperref}
}
\longversion{
\usepackage[colorlinks=true,linkcolor=blue,citecolor=Mahogany]{hyperref}
}
\usepackage{makecell}
\usepackage{multirow}
\longversion{\usepackage{eulervm}}
\usepackage{graphicx}
\usepackage{adjustbox}
\usepackage{anyfontsize}

\longversion{
\usepackage{charter}
\usepackage{fullpage}
}

\shortversion{
\floatsep 2pt plus 1pt minus 1pt
\textfloatsep 2pt plus 1pt minus 1pt
\intextsep 2pt plus 1pt minus 1pt
\dblfloatsep 2pt plus 1pt minus 1pt
\dbltextfloatsep 2pt plus 1pt minus 1pt

\setlength{\belowdisplayskip}{2pt} \setlength{\belowdisplayshortskip}{0pt}
\setlength{\abovedisplayskip}{2pt} \setlength{\abovedisplayshortskip}{0pt}
}

\usepackage{url}

\usepackage{enumerate}

\usepackage[shortlabels]{enumitem}

\usepackage{xspace}

\newcommand{\suc}{\ensuremath{\succ}\xspace}
\let\mydelta\delta
\renewcommand{\delta}{\ensuremath{\mydelta}\xspace}
\let\myalpha\alpha
\renewcommand{\alpha}{\ensuremath{\myalpha}\xspace}
\let\mystar\star
\renewcommand{\star}{\ensuremath{\mystar}}

\newcommand{\NP}{\ensuremath{\mathsf{NP}}\xspace}
\newcommand{\NPC}{\ensuremath{\mathsf{NP}}-complete\xspace}
\newcommand{\el}{\ensuremath{\ell}\xspace}

\newcommand{\Pb}{\ensuremath{P}\xspace}
\newcommand{\TSAT}{\ensuremath{(3,\text{B}2)}-{\sc SAT}\xspace}

\newcommand{\YES}{{\sc yes}\xspace}
\newcommand{\NO}{{\sc no}\xspace}

\newcommand{\OB}{\textsc{Local Distance constrained bribery}\xspace}
\newcommand{\ODB}{\textsc{Local Distance constrained \$bribery}\xspace}

\newcommand{\SHB}{\textsc{Shift bribery}\xspace}
\newcommand{\SWB}{\textsc{Swap bribery}\xspace}
\newcommand{\DB}{\$\textsc{Bribery}\xspace}
\newcommand{\FB}{\textsc{Frugal bribery}\xspace}

\renewcommand{\AA}{\ensuremath{\mathcal A}\xspace}
\newcommand{\BB}{\ensuremath{\mathcal B}\xspace}
\newcommand{\CC}{\ensuremath{\mathcal C}\xspace}
\newcommand{\DD}{\ensuremath{\mathcal D}\xspace}
\newcommand{\EE}{\ensuremath{\mathcal E}\xspace}

\newcommand{\GG}{\ensuremath{\mathcal G}\xspace}

\newcommand{\LL}{\ensuremath{\mathcal L}\xspace}

\newcommand{\NN}{\ensuremath{\mathcal N}\xspace}

\newcommand{\PP}{\ensuremath{\mathcal P}\xspace}
\newcommand{\QQ}{\ensuremath{\mathcal Q}\xspace}
\newcommand{\RR}{\ensuremath{\mathcal R}\xspace}

\newcommand{\VV}{\ensuremath{\mathcal V}\xspace}
\newcommand{\WW}{\ensuremath{\mathcal W}\xspace}
\newcommand{\XX}{\ensuremath{\mathcal X}\xspace}

\newcommand{\NB}{\ensuremath{\mathbb N}\xspace}
\newcommand{\ZB}{\ensuremath{\mathbb Z}\xspace}

\usepackage{nicefrac}

\newcommand{\nfrac}{\nicefrac}

\newcommand{\eps}{\ensuremath{\varepsilon}\xspace}
\renewcommand{\epsilon}{\eps}

\newcommand{\ignore}[1]{}

\newcommand{\pr}{\ensuremath{\prime}}
\newcommand{\prr}{\ensuremath{{\prime\prime}}}

\renewcommand{\ge}{\geqslant}
\renewcommand{\le}{\leqslant}

\usepackage{cleveref}

\crefname{theorem}{theorem}{\bf Theorems}
\crefname{observation}{observation}{\bf Observations}
\crefname{lemma}{lemma}{\bf Lemmas}
\crefname{corollary}{corollary}{\bf Corollaries}
\crefname{proposition}{proposition}{\bf Propositions}
\crefname{definition}{definition}{\bf Definitions}
\crefname{claim}{claim}{\bf Claims}
\crefname{reductionrule}{reduction rule}{\bf Reduction rules}

\longversion{\newtheorem{theorem}{\bf Theorem}}
\longversion{}

\longversion{\newtheorem{lemma}[theorem]{\bf Lemma}}
\longversion{}
\longversion{\newtheorem{definition}[theorem]{\bf Definition}}

%% file: abstract.tex
\begin{abstract}
 Studying complexity of various bribery problems has been one of the main research focus in computational social choice. In all the models of bribery studied so far, the briber has to pay every voter some amount of money depending on what the briber wants the voter to report and the briber has some budget at her disposal. Although these models successfully capture many real world applications, in many other scenarios, the voters may be unwilling to deviate too much from their true preferences. In this paper, we study the computational complexity of the problem of finding a preference profile which is as close to the true preference profile as possible and still achieves the briber's goal subject to budget constraints. We call this problem \ODB. We consider three important measures of distances, namely, swap distance, footrule distance, and maximum displacement distance, and resolve the complexity of the optimal bribery problem for many common voting rules. We show that the problem is polynomial time solvable for the plurality and veto voting rules for all the three measures of distance. On the other hand, we prove that the problem is \NPC for a class of scoring rules which includes the Borda voting rule, maximin, Copeland$^\alpha$ for any $\alpha\in[0,1]$, and Bucklin voting rules for all the three measures of distance even when the distance allowed per voter is $1$ for the swap and maximum displacement distances and $2$ for the footrule distance even without the budget constraints (which corresponds to having an infinite budget). For the $k$-approval voting rule for any constant $k>1$ and the simplified Bucklin voting rule, we show that the problem is \NPC for the swap distance even when the distance allowed is $2$ and for the footrule distance even when the distance allowed is $4$ even without the budget constraints.\longversion{ We complement these hardness results by showing that the problem for the $k$-approval and simplified Bucklin voting rules is polynomial time solvable for the swap distance if the distance allowed is $1$ and for the footrule distance if the distance allowed is at most $3$.} For the $k$-approval voting rule for the maximum displacement distance for any constant $k>1$, and for the simplified Bucklin voting rule for the maximum displacement distance, we show that the problem is \NPC (with the budget constraints) and, without the budget constraints, they are polynomial time solvable.

\end{abstract}

%% file: introduction.tex
\section{Introduction}

\begin{table*}
\centering
\begin{adjustbox}{max width=\textwidth}
\shortversion{\renewcommand{\arraystretch}{1.5}}
\longversion{\renewcommand{\arraystretch}{2.5}}
 \begin{tabular}{|c|c|c|c|}\hline
  \multirow{2}{*}{Voting rule} & \multicolumn{3}{c|}{Distance Metric}\\\cline{2-4}
  &Swap & Footrule & Maximum displacement\\\hline\hline
  
  Plurality & \multicolumn{3}{c|}{\multirow{2}{*}{$^\star$\Pb [\Cref{thm:plurality_poly}]}}\\\cline{1-1}
  
  Veto & \multicolumn{3}{c|}{}\\\hline
  
  $k$-approval & \makecell{$^\star$\Pb for $\delta=1$ [\Cref{thm:kapp_poly}]\\ \NPC for $\delta= 2$ [\Cref{thm:ob_kapp_swap}]} & \makecell{$^\star$\Pb for $\delta\le 3$ [\Cref{thm:kapp_poly}]\\ \NPC for $\delta= 4$ [\Cref{thm:ob_kapp_swap}]} & \makecell{\Pb [\Cref{thm:kapp_maxdis_poly}], $^\star$\Pb for $\delta_i=1, \forall i$ [\Cref{thm:kapp_poly}]\\$^\star$\NPC for $\delta_i=2, \forall i$ [\Cref{thm:kapp_max_hard}]} \\\hline
  
  Borda & \makecell{\NPC for $\delta= 1$ [\Cref{thm:borda_max}]} & \makecell{\NPC for $\delta=2$ [\Cref{thm:borda_max}]}& \makecell{\NPC for $\delta=1$ [\Cref{thm:borda_max}]}\\\hline
  
  Maximin & \makecell{\NPC for $\delta= 1$ [\Cref{thm:ob_maximin_swap}]} & \makecell{\NPC for $\delta= 2$ [\Cref{thm:ob_maximin_swap}]}& \makecell{\NPC for $\delta= 1$ [\Cref{thm:ob_maximin_swap}]} \\\hline
  
  \makecell{Copeland$^\alpha, \alpha\in[0,1]$} &\makecell{\NPC for $\delta= 1$ [\Cref{thm:ob_copeland_swap}]} & \makecell{\NPC for $\delta= 2$ [\Cref{thm:ob_copeland_swap}]}& \makecell{\NPC for $\delta= 1$ [\Cref{thm:ob_copeland_swap}]}\\\hline

  Simplified Bucklin & \makecell{$^\star$\Pb for $\delta=1$  [\Cref{thm:bucklin_poly}]\\ \NPC for $\delta= 2$ [\Cref{thm:ob_bucklin_swap}]} & \makecell{$^\star$\Pb for $\delta\le 3$  [\Cref{thm:bucklin_poly}]\\ \NPC for $\delta= 4$ [\Cref{thm:ob_bucklin_swap}]} & \makecell{\Pb [\Cref{thm:bucklin_maxdis_poly}], $^\star$\Pb for $\delta_i=1, \forall i$ [\Cref{thm:kapp_poly}]\\$^\star$\NPC for $\delta_i=2, \forall i$ [\Cref{thm:ob_bucklin_maxdis_hard}]} \\\hline

  Bucklin & \makecell{\NPC for $\delta= 1$ [\Cref{thm:ob_bucklin}]} & \makecell{\NPC for $\delta=2$ [\Cref{thm:ob_bucklin}]}& \makecell{\NPC for $\delta=1$ [\Cref{thm:ob_bucklin}]}\\\hline
 \end{tabular}
\end{adjustbox}
 \caption{The results marked $\star$ hold for the \ODB problem; others hold for the \OB problem.}\label{tbl:summary}
\end{table*}

Aggregating preferences of a set of agents over a set of alternatives is a fundamental problem in many applications both in real life and artificial intelligence\shortversion{~\cite{Cohen,PennockHG00}}.\longversion{ Voting has often served as a natural tool for aggregating preferences in such applications. Pioneering use of voting in key applications of artificial intelligence includes spam detection~\cite{Cohen}, collaborative filtering~\cite{PennockHG00}, etc.} A typical voting setting consists of a set of alternatives, a set of agents each having a preference which is a complete order over the alternatives, and a voting rule which declares one or more alternatives as the winner(s) of the election.

However any such election scenario is susceptible to control attacks of various kinds -- internal or external agents may try to influence the election system in someone's favor. One such attack which has been studied extensively in computational social choice is {\em bribery}. In every model of bribery studied so far (see~\cite{DBLP:reference/choice/FaliszewskiR16}), we have the preferences of a set of voters, an external agent called briber with some budget, a bribing model which dictates how much one has to bribe any voter to persuade her to cast a vote of briber's choice, and the computational problem is to check whether it is possible to bribe the voters subject to the budget constraint so that some alternative of briber's choice becomes the winner. This models not only serve as a true theoretical abstraction of various real world scenarios but also generalizes many other important control attacks, for example, coalitional manipulation~\cite{bartholdi1989computational,conitzer2007elections}. In this paper, we study a refinement of the above bribery model motivated by the following important observation made by Obraztsova and Elkind~\cite{DBLP:conf/aaai/ObraztsovaE12,DBLP:conf/aamas/ObraztsovaE12}

\vspace{-1ex}
\begin{quote}
{\em ``...if voting is public (or
if there is a risk of information leakage), and a voter's preference is at least somewhat known to her friends and colleagues, she may be worried that voting non-truthfully can harm her reputation yet hope that she will not be caught if her vote is sufficiently similar to her true ranking. Alternatively, a voter who is uncomfortable about manipulating an election for ethical reasons may find a lie more palatable if it does not require her to re-order more than a few candidates.''}
\end{quote}

Indeed, in the context of bribery, there can be situations where a voter may be bribed to report some preference which ``resembles'' her true preference but a voter is simply unwilling to report any preference which is far from her true preference. We remark that existing models of bribery do not capture the above constraint since, intuitively speaking, the budget feasibility constraint in these models restricts the total money spent (which is a global constraint) whereas the situations above demand (local) constraints per voter. For example, let us think of a voter $v$ with preference $a\suc b\suc c$. Suppose the voter $v$ can be persuaded to make at most two swaps and the cost of persuading her does not depend on the number of swaps she performs in her preference. This could be the situation when she is happy to change her preference as briber advises (simply because she trusts the briber that her change will finally ensure a better social outcome) but does not wish to deviate from her own preference too much to avoid social embarrassment. One can see that the classical model of bribery (\SWB for example) fails to capture the intricacies of this situation (for example, making the cost per swap to be $0$ fails because the voter $v$ is not willing to cast $c\suc b\suc a$). In this paper, we fill this research gap by proposing a bribery model which directly addresses these scenarios.

More specifically, we study the computational complexity of the following problem which we call \ODB. Given preferences $\PP=(\suc_i)_{i\in[n]}$ of a set of agents, non-negative integers $(\delta_i)_{i\in[n]}$ denoting the distance change allowed for corresponding agents, non-negative integers $(p_i)_{i\in[n]}$ where $p_i$ denotes the amount one has to pay the voter $i$ to change his or her preference, a non-negative integer budget $\BB$, and an alternative $c$, compute if the preferences can be changed subject to the ``price, distance, and budget constraints'' so that $c$ is a winner in the resulting election for some voting rule. We also study an interesting special case of the \ODB problem where $\delta_i=\delta$ for some non-negative integer \delta and $p_i=0$ for every $i$ and $\BB=0$; we call the latter problem \OB. In this paper, we study the following commonly used distance functions on the set of all possible preferences (permutations on the set of alternatives): (i) swap distance~\cite{kendall1938new}, (ii) footrule distance~\cite{spearman1904proof}, and (iii) maximum displacement distance~\shortversion{\shortcite{DBLP:conf/aaai/ObraztsovaE12,DBLP:conf/aamas/ObraztsovaE12}}\longversion{\cite{DBLP:conf/aaai/ObraztsovaE12,DBLP:conf/aamas/ObraztsovaE12}}. The swap distance (aka Kendall Tau distance, bubble sort distance, etc.) between two preferences is the number of pairs of alternatives which are ranked in different order in these two preferences. Whereas the footrule distance (maximum displacement distance respectively) between two preferences is the sum (maximum respectively) of the absolute value of the differences of the positions of every alternative in two preferences. We refer to \Cref{sec:prelim} for formal definitions of the problems and the distance functions above.

\subsection{Contribution}

We study the computational complexity of the \ODB and \OB problems for the plurality, veto, $k$-approval, a class of scoring rules which includes the Borda voting rule [\Cref{cor:scr}], maximin, Copeland$^\alpha$ for any $\alpha\in[0,1]$, Bucklin, and simplified Bucklin voting rules for the swap, footrule, and maximum displacement distance. We summarize our results in \Cref{tbl:summary}. We highlight that all our results are tight in the sense that we even find the exact value of $\delta$ till which the \OB problem for some particular voting rule is polynomial time solvable and beyond which it is \NPC. As can be observed in \Cref{tbl:summary}, most of our \NP-completeness results (except \Cref{thm:kapp_max_hard,thm:ob_bucklin_maxdis_hard} for which the corresponding \OB problems are polynomial time solvable) hold even for the \OB problem (and thus for the \ODB problem too) and even for small constant values for \delta. On the other hand, most of our polynomial time algorithms (except \Cref{thm:kapp_poly,thm:bucklin_maxdis_poly} for which the corresponding \ODB problems are \NPC) work for the general \ODB problem (and thus for the \OB problem too). We would like to highlight a curious case -- for the maximum displacement distance, the \OB problem is polynomial time solvable for the simplified Bucklin voting rule (for any \delta) and \NPC for the Bucklin voting rule even for $\delta=1$. To the best of our knowledge, this is the first instance where a natural problem is polynomial time solvable for the simplified Bucklin voting rule and \NPC for the Bucklin voting rule. We also observe that, unlike the optimal manipulation problem in~\shortversion{\shortcite{DBLP:conf/aaai/ObraztsovaE12,DBLP:conf/aamas/ObraztsovaE12}}\longversion{\cite{DBLP:conf/aaai/ObraztsovaE12,DBLP:conf/aamas/ObraztsovaE12}}, the complexity of the \OB problem for some common voting rule ($k$-approval with $k>1$ and simplified Bucklin voting rules for example) can depend significantly on the distance function under consideration.

\subsection{Related Work}

Faliszewski et al.~\longversion{\cite{faliszewski2006complexity}}\shortversion{\shortcite{faliszewski2006complexity}} propose the first bribery problem where the briber's goal is to change a minimum number of preferences to make some candidates win the election. Then they extend their basic model to more sophisticated models of \SHB and \DB~\longversion{\cite{faliszewski2009hard,faliszewski2009llull}}\shortversion{\shortcite{faliszewski2009hard,faliszewski2009llull}}. Elkind et al.~\longversion{\cite{elkind2009swap}}\shortversion{\shortcite{elkind2009swap}} extend this model further and study the \SWB problem where there is a cost associated with every swap of alternatives. Dey et al.~\longversion{\cite{frugalDeyMN16}}\shortversion{\shortcite{frugalDeyMN16}} show that the bribery problem remains intractable for many common voting rules for an interesting special case which they call \FB. The bribery problem has also been studied in various other preference models, for example, truncated ballots~\cite{BaumeisterFLR12}, soft constraints~\cite{pini2013bribery}, approval ballots~\cite{schlotter2017campaign}, campaigning in societies~\cite{faliszewski2018opinion}, CP-nets~\cite{dorn2014hardness}, combinatorial domains~\cite{mattei2012bribery}, iterative elections~\cite{maushagen2018complexity}, committee selection~\cite{bredereck2016complexity}, probabilistic lobbying~\cite{erdelyi2009complexity}, etc. Erdelyi et al.~\cite{erdelyi2014bribery} study the bribery problem under voting rule uncertainty. Faliszewski et al.~\longversion{\cite{faliszewski2014complexity}}\shortversion{\shortcite{faliszewski2014complexity}} 
study bribery for the simplified Bucklin and the Fallback voting rules. Xia~\longversion{\cite{xia2012computing}}\shortversion{\shortcite{xia2012computing}}, and  Kaczmarczyk and Faliszewski~\longversion{\cite{kaczmarczyk2016algorithms}}\shortversion{\shortcite{kaczmarczyk2016algorithms}} study the destructive variant of bribery. Dorn and Schlotter~\longversion{\cite{dorn2012multivariate}}\shortversion{\shortcite{dorn2012multivariate}} and Bredereck et al.~\longversion{\cite{bredereck2014prices}}\shortversion{\shortcite{bredereck2014prices}} explore parameterized complexity of various bribery problems. Chen et al.~\longversion{\cite{chen2018protecting}}\shortversion{\shortcite{chen2018protecting}} provide novel mechanisms to protect elections from bribery. Knop et al.~\longversion{\cite{Knop}}\shortversion{\shortcite{Knop}} provide a uniform framework for various control problems. Although most of the bribery problems are intractable, few of them, \SHB for example, have polynomial time approximation algorithms~\cite{elkind2010approximation,DBLP:conf/aaai/KellerHH18}. Manipulation, a specialization of bribery, is another fundamental attack on election~\cite{DBLP:reference/choice/ConitzerW16}. In the manipulation problem, a set of voters (called manipulators) wants to cast their preferences in such a way that (when tallied with the preferences of other preferences) makes some alternative win the election. Obraztsova and Elkind~\shortversion{\shortcite{DBLP:conf/aaai/ObraztsovaE12,DBLP:conf/aamas/ObraztsovaE12}}\longversion{\cite{DBLP:conf/aaai/ObraztsovaE12,DBLP:conf/aamas/ObraztsovaE12}} initiate the study of optimal manipulation in that context.

The rest of the paper has been organized as follows. We introduce our notation and formally define our computational problems in \Cref{sec:prelim}; then we present our polynomial time algorithms and \NP-completeness results in respectively \Cref{sec:poly} and \Cref{sec:hard}; we finally conclude with future research directions in \Cref{sec:con}. A short version of our work will appear in~\cite{optbribery2019}.

%% file: prelim.tex
\section{Preliminaries}\label{sec:prelim}

\longversion{\subsection{Voting and Voting Rules}}

For a positive integer $k$, we denote the set $\{1, 2, \ldots, k\}$ by $[k]$. Let $\AA=\{a_i: i\in[m]\}$ be a set of $m$ {\em alternatives}. A complete order over the set \AA of alternatives is called a {\em preference}. We denote the set of all possible preferences over \AA by $\LL(\AA)$. A tuple $(\suc_i)_{i\in[n]}\in\LL(\AA)^n$ of $n$ preferences is called a {\em profile}. We say that an alternative $a\in\AA$ is placed at the $\el^{th}$ position (from left or from top) in a preference $\suc\in\LL(\AA)$ for some positive integer \el (and denote it by $\text{pos}(a,\suc)$) if $|\{b\in\AA: b\suc a\}|=\el-1$. We say that the distance of two alternatives $a,b\in\AA$ in a preference $\suc\in\LL(\AA)$ is some positive integer $k$ if there exists a positive integer \el such that the positions of $a$ and $b$ in \suc are either $\el$ and $\el+k$ respectively or $\el+k$ and \el respectively. An {\em election} \EE is a tuple $(\suc,\AA)$ where \suc is a profile over a set \AA of alternatives. If not mentioned otherwise, we denote the number of alternatives and the number of preferences by $m$ and $n$ respectively. A map $r:\uplus_{n,|\mathcal{A}|\in\mathbb{N}^+}\mathcal{L(A)}^n\longrightarrow 2^\mathcal{A}\setminus\{\emptyset\}$ is called a \emph{voting rule}. Given an election $\EE$, we can construct from \EE a directed weighted graph $\GG_\EE$ which is called the \textit{weighted majority graph} of \EE. The set of vertices in $\GG_\EE$ is the set of alternatives in $\EE$. For any two alternatives $x$ and $y$, the weight of the edge $(x,y)$ is $\DD_\EE(x,y) = \NN_\EE(x,y) - \NN_\EE(y,x)$, where $\NN_\EE(a, b)$ is the number of preferences where the alternative $a$ is preferred over the alternative $b$ for $a,b\in\AA, a\ne b$. Examples of some common voting rules are as follows.

\begin{itemize}[leftmargin=0cm,itemindent=.5cm,labelwidth=\itemindent,labelsep=0cm,align=left]
 \item \textbf{Positional scoring rules:} An $m$-dimensional vector $\alpha=\left(\alpha_1,\alpha_2,\dots,\alpha_m\right)\in\mathbb{N}^m$ with $\alpha_1\ge\alpha_2\ge\dots\ge\alpha_m$ and $\alpha_1>\alpha_m$ for every $m\in \mathbb{N}$ naturally defines a voting rule --- an alternative gets score $\alpha_i$ from a preference if it is placed at the $i^{th}$ position, and the score of an alternative is the sum of the scores it receives from all the preferences. The winners are the alternatives with the maximum score. \longversion{Scoring rules remain unchanged if we multiply every $\alpha_i$ by any constant $\lambda>0$ and/or add any constant $\mu$. Hence, we can assume without loss of generality that for any score vector $\alpha$, we have $gcd((\alpha_i)_{i\in[m]})=1$ and there exists a $j< m$ such that $\alpha_\el = 0$ for all $\el>j$. We call such an $\alpha$ a {\em normalized} score vector.} For some $k\in[m-1]$, if $\alpha_i$ is $1$ for $i\in [k]$ and $0$ otherwise, then, we get the {\em $k$-approval} voting rule. \longversion{The {\em $k$-approval} voting rule is also called the {\em $(m-k)$-veto} voting rule.} The $1$-approval voting rule is called the {\em plurality} voting rule and the $(m-1)$-approval voting rule is called the {\em veto} voting rule. If $\alpha_i=m-i$ for every $i\in[m]$, then we get the {\em Borda} voting rule.
 
 \item \textbf{Maximin:} The maximin score of an alternative $x$ is $\min_{y\ne x} \DD_\EE(x,y)$. The winners are the alternatives with the maximum maximin score.
 
 \item \textbf{Copeland$^{\alpha}$:} Given $\alpha\in[0,1]$, the Copeland$^{\alpha}$ score of an alternative $x$ is $|\{y\ne x:\DD_\EE(x,y)>0\}|+\alpha|\{y\ne x:\DD_\EE(x,y)=0\}|$. The winners are the alternatives with the maximum Copeland$^{\alpha}$ score. 
 
 \item \textbf{Simplified Bucklin and Bucklin:} The simplified Bucklin score of an alternative $x$ is the minimum number $\ell$ such that $x$ is placed within the first $\ell$ positions in more than half of the preferences. The winners are the alternatives with the lowest simplified Bucklin score. Let $k$ be the minimum simplified Bucklin score of any alternative. Then the Bucklin winners are the alternatives who appear the maximum number of times within the first $k$ positions.
\end{itemize}

\longversion{\subsection{Distance Function}}

A distance function $d$ takes two preferences on a set \AA of alternatives as input and outputs a non-negative number which satisfies the following properties: (i) $d(\suc_1,\suc_2)=0$ for $\suc_1,\suc_2\in\LL(\AA)$ if and only if $\suc_1=\suc_2$, (ii) $d(\suc_1,\suc_2)=d(\suc_2,\suc_1)$ for every $\suc_1, \suc_2\in\LL(\AA)$, (iii) $d(\suc_1,\suc_2)\le d(\suc_1,\suc_3)+d(\suc_3,\suc_2)$ for every $\suc_1, \suc_2, \suc_3\in\LL(\AA)$. We will consider the following distance functions in this paper. 
\begin{itemize}
 \item {\bf Swap distance:} $$d_{\text{swap}}(\suc_1,\suc_2)=\big|\{\{a,b\}\subset\AA: a\suc_1 b, b\suc_2 a\}\big|$$
 \item {\bf Footrule distance:} $$d_{\text{footrule}}(\suc_1,\suc_2)=\sum_{a\in\AA} \Big|\text{pos}(a,\suc_1)-\text{pos}(a,\suc_2)\Big|$$
 \item {\bf Maximum displacement distance:} $$d_{\text{max displacement}}(\suc_1,\suc_2)=\max_{a\in\AA} \Big|\text{pos}(a,\suc_1)-\text{pos}(a,\suc_2)\Big|$$
\end{itemize}

It is well known that, for any two preferences $\suc_1, \suc_2\in\LL(\AA)$, we have $d_{\text{swap}}(\suc_1,\suc_2)\le d_{\text{footrule}}(\suc_1,\suc_2)\le 2 d_{\text{swap}}(\suc_1,\suc_2)$~\cite{diaconis1977spearman}. Let $r$ be any voting rule and $d$ be any distance function on $\LL(\AA)$. We now define our computational problem formally for any distance function $d$ on $\LL(\AA)$.

 \noindent\fbox{\begin{minipage}{\linewidth}
\begin{definition}[$d$-\OB]\label{def:prob}
 Given a set \AA of alternatives, a profile $\suc=(\suc_i)_{i\in[n]}\in\LL(\AA)^n$ of $n$ preferences, a positive integer $\delta$, and an alternative $c\in\AA$, compute if there exists a profile $\suc^\pr=(\suc^\pr_i)_{i\in[n]}\in\LL(\AA)^n$ such that
 \begin{enumerate}[(i),leftmargin=*]
  \item $d(\suc_i,\suc^\pr_i) \le \delta$ for every $i\in[n]$
  \item $r(\suc^\pr)=\{c\}$
 \end{enumerate} 
 We denote an instance of \OB by $(\AA,\PP,c,\delta)$.
\end{definition}
 \end{minipage}}

\longversion{\vspace{2ex}}
 \noindent\fbox{\begin{minipage}{\linewidth}
\begin{definition}[$d$-\ODB]\label{def:prob_dollar}
 Given a set \AA of alternatives, a profile $\suc=(\suc_i)_{i\in[n]}\in\LL(\AA)^n$ of $n$ preferences, positive integers $(\delta_i)_{i\in[n]}$ denoting distances allowed for every preference, non-negative integers $(p_i)_{i\in[n]}$ denoting the prices of every preference, a non-negative integer \BB denoting the budget of the Briber, and an alternative $c\in\AA$, compute if there exists a subset $J\subseteq[n]$ and a profile $\suc^\pr=(\suc^\pr_i)_{i\in J}\in\LL(\AA)^{|J|}$ such that
 \begin{enumerate}[(i),leftmargin=*]
  \item $\sum_{i\in J} p_i \le \BB$
  \item $d(\suc_i,\suc^\pr_i) \le \delta_i$ for every $i\in J$
  \item $r\left((\suc^\pr_i)_{i\in J}, (\suc_i)_{i\in[n]\setminus J}\right)=\{c\}$
 \end{enumerate} 
 We denote an instance of \ODB by $(\AA,\PP,c,(\delta_i)_{i\in[n]},(p_i)_{i\in[n]})$.
\end{definition}
 \end{minipage}}

We remark that the optimal bribery problem, as described in \Cref{def:prob}, demands the alternative $c$ to win uniquely. It is equally motivating to demand that $c$ is a co-winner. As far as the optimal bribery problem is concerned, we can easily verify that all our results, both algorithmic and hardness, extend easily to the co-winner case. However, we note that it need not always be the case in general (see Section 1.1 in~\cite{DBLP:journals/jair/XiaC11} for example).

%% file: results.tex
\section{Polynomial Time Algorithms}\label{sec:poly}

In this section, we present our polynomial time algorithms. All our polynomial time algorithms are obtained by reducing our problem to the maximum flow problem which is quite common in computational social choice (see for example \cite{faliszewski2008nonuniform}). In the maximum flow problem, the input is a directed graph with two special vertices $s$ and $t$ and (positive) capacity for every edge, and the goal is to compute the value of maximum flow that can be sent from $s$ to $t$ subject to the capacity constraints of every edge. This problem is known to have polynomial time algorithms. Moreover, it is also known that, if the capacity of every edge is positive integer, then the value of maximum flow is also a positive integer, the flow in every edge is also a non-negative integer, and such a maximum flow can be found in polynomial time~\cite{cormen2009introduction}.

\begin{theorem}\label{thm:plurality_poly}
 The \ODB problem is polynomial time solvable for the plurality and veto voting rules for the swap, footrule, and maximum displacement distance.
\end{theorem}

\begin{proof}
 Let us first prove the result for the plurality voting rule. Let $\PP=(\suc_i)_{i\in[n]}$ be the input profile, and $x$ the distinguished alternative. Let the plurality score of any alternative $a\in\AA$ in \PP be $s(a)$. We first guess the final score $\el_x$ of $x$ in the range $s(x)$ to $n$. Let \QQ be the sub-profile of \PP consisting of preferences which do not place $x$ at their first position. For any preference $\suc_i\in\QQ$, we compute the set of alternatives $\AA_i\subseteq\AA$ which can be placed at the first position keeping the distance from $\suc_i$ at most $\delta_i$. We observe that, for the swap and maximum displacement distances, $\AA_i$ is the set of alternatives which appear within the first $\delta_i+1$ positions in $\suc_i$; for the footrule distance, $\AA_i$ is the set of alternatives which appear within the first $\lfloor\nfrac{\delta_i}{2}\rfloor$ 
 positions in $\suc_i$. We now create the following minimum cost flow network \GG with demand on edges. It is well known that a minimum cost flow (of certain flow value) satisfying demands of all the edges can be found in polynomial time~(see for example \longversion{\cite{flow}}\shortversion{\shortcite{flow}}). 
\longversion{\vspace{2ex}}

 \noindent\fbox{\begin{minipage}{\linewidth}
 \begin{gather*}
  \VV[\GG]=\{u_i: \suc_i\in\QQ\} \cup \{v_a: a\in\AA\} \cup \{s,t\}\\
  \EE[\GG]=\{(s,u_i): \suc_i\in\QQ\} \cup \{(v_a,t): a\in\AA\}\\
  \cup \{(u_i,v_a): \suc_i\in\QQ, a\in\AA_i\}\\
  \text{capacity}(v_x,t)=\el_x-s(x);\text{capacity}(v_a,t)=\el_x-1 \forall a\in\AA\setminus\{x\}\\ \text{capacity of every other edge is }1\\
  \text{demand}(v_x,t)=\el_x-s(x); \text{demand of every other edge is }0\\
  \text{cost}(u_i,v_a)=p_i \forall \suc_i\in\QQ, a\in\AA_i, a\text{ does not appear}\\\text{at the first position in }\suc_i;\text{ cost of other edges is }0
 \end{gather*}
 \end{minipage}}
\longversion{\vspace{2ex}}
 
 From the construction it follows that the input \ODB instance has a successful bribery where the final plurality score of the alternative $x$ is at least $\el_x$ if and only if there is an $s-t$ flow in \GG of flow value $|\QQ|$ with cost at most \BB. Since there are at most $n$ possible values that $\el_x$ can take and we try all of them, the \ODB problem for the plurality voting rule for the three distance functions is polynomial time solvable.
 
 We now prove the result for the veto voting rule. We first guess the final veto score $\el_x$ of $x$ in the range $0$ to $n$. For any preference $\suc_i\in\PP$, we compute the set of alternatives $\AA_i\subseteq\AA$ which can be placed at the last position keeping the distance from $\suc_i$ at most $\delta_i$. We observe that, for all the three the distance functions under consideration, $\AA_i$ is the set of all alternatives which appear within the last $\delta_i$ positions in $\suc_i$.  We now create the following minimum cost flow network \GG with demand on edges. 
 \longversion{\vspace{2ex}}
 
 \noindent\fbox{\begin{minipage}{\linewidth}
 \begin{gather*}
  \VV[\GG]=\{u_i: \suc_i\in\PP\} \cup \{v_a: a\in\AA\} \cup \{s,t\}\\
  \EE[\GG]=\{(s,u_i): \suc_i\in\PP\} \cup \{(v_a,t): a\in\AA\}\\
  \cup \{(u_i,v_a): \suc_i\in\PP, a\in\AA_i\}\\
  \text{capacity}(v_x,t)=\el_x;\text{capacity}(v_a,t)=n \forall a\in\AA\setminus\{x\};\\ \text{capacity of every other edge is }1\\
  \text{demand}(v_a,t)=\el_x+1\forall a\in\AA\setminus\{x\}; \text{demand of every other edge is }0\\
  \text{cost}(u_i,v_a)=p_i \forall \suc_i\in\PP, a\in\AA_i, a\text{ does not appear}\\\text{at the last position in }\suc_i;\text{ cost of other edges is }0
 \end{gather*}
 \end{minipage}}
\longversion{\vspace{2ex}}
 
 From the construction it follows that the input \ODB instance has a successful bribery where the final veto score of the alternative $x$ is at most $\el_x$ if and only if there is an $s-t$ flow in \GG of flow value $|\PP|$ with cost at most \BB. Since there are at most $n$ possible values that $\el_x$ can take and we try all of them, the \ODB problem for the veto voting rule for the three distance functions is polynomial time solvable.
\end{proof}

We prove the following results by reducing to the maximum flow problem.

\begin{theorem}\label{thm:kapp_poly}2
 The \ODB problem is polynomial time solvable for the $k$-approval voting rule for any $k$ for the swap distance distance and the maximum displacement distance when $\delta_i=1$ for every preference $i$ and thus for the footrule distance when $\delta_i\le 3$ for every preference $i$.
\end{theorem}

\longversion{
\begin{proof}
 The proof for the swap and maximum displacement distance is completely analogous to the proof of \Cref{thm:plurality_poly}. The second part of the statement follows from the following claim: for any two preferences $\suc_1,\suc_2\in\LL(\AA)$, if $d_{\text{footrule}}(\suc_1,\suc_2)\le 3$, then $d_{\text{swap}}(\suc_1,\suc_2)\le 1$. If $d_{\text{footrule}}(\suc_1,\suc_2)= 2$, then clearly $d_{\text{swap}}(\suc_1,\suc_2)=1$. Now if, $d_{\text{footrule}}(\suc_1,\suc_2)= 3$, then $d_{\text{swap}}(\suc_1,\suc_2)\ge 1$. However, if $d_{\text{swap}}(\suc_1,\suc_2)=2$, then, by analyzing all the cases, $d_{\text{footrule}}(\suc_1,\suc_2)= 4$ which proves the claims.
\end{proof}
}

\begin{theorem}\label{thm:kapp_maxdis_poly}\shortversion{[\star]}
 The \OB problem is polynomial time solvable for the $k$-approval voting rule for the maximum displacement distance for any $k$.
\end{theorem}

\longversion{
\begin{proof}
 Let \PP be the input profile and $x$ the distinguished alternative. In every preference in \PP, if it is possible to place $x$ within the first $k$ position, which happens exactly when $x$ appears within the first $k+\delta$ positions in a preference, we place $x$ at the first position and keep the relative ordering of every other alternative the same. Let $\PP^\pr\subseteq\PP$ be the set of such preferences. Then the score of $x$ in the resulting profile is $|\PP^\pr|$; let it be $\el_x$. For any preference $\suc\in\PP$, we compute the set of alternatives $\AA_\suc\subseteq\AA$ for whom there exist two preferences $\suc_1, \suc_2\in\LL(\AA)$ with distance at most \delta each from \suc such that the alternative appears within the first $k$ positions in $\suc_1$ and it does not appear within the first $k$ positions in $\suc_2$. We observe that $\AA_\suc$ is precisely the set of alternatives which appear in positions from $\max\{1,k-\delta\}$ to $\min\{k+\delta,m\}$. Let $\delta^\pr = \min\{k,\delta\}$. For any alternative $y\in\AA\setminus\{x\}$, let $\el_y$ be the number of preferences in \PP where $y$ appears within the first $k-\delta-1$ positions. If there exists an alternative $y\in\AA\setminus\{x\}$ such that $\el_y\ge \el_x$, then we output \NO. Otherwise we create the following flow network
 
 \vspace{2ex}
 \noindent\fbox{\begin{minipage}{\linewidth}
 \begin{align*}
 	\GG &= (\VV,\EE,c,s,t) \text{ where}\\
 	\VV &= \{u_\suc:\suc\in\PP\} \cup \{v_y: y\in\AA\setminus\{x\}\} \cup \{s,t\}\\
 	\EE &= \{(s,u_\suc): \suc\in\PP\} \cup \{(v_y,t): y\in\AA\setminus\{x\}\} \cup \{(u_\suc,v_y): \suc\in\PP, y\in\AA\setminus\{x\}, y\in\AA_\suc\}\\
 	c((v_y,t)) &=\el_x-1-\el_y\; \forall (v_y,t)\in\EE\\
 	c((s,u_\suc)) &=\delta^\pr-1\; \forall (s,u_\suc)\in\EE, \suc\in\PP^\pr\\
 	c((s,u_\suc)) &=\delta^\pr\; \forall (s,u_\suc)\in\EE, \suc\in\PP\setminus\PP^\pr\\
 	c(e) &= 1\;\text{ for every other edge}
 \end{align*}
 \end{minipage}}\vspace{2ex}

 From the construction it is clear that the input instance is a \YES instance if and only if there is an $s-t$ flow in \GG of value $\sum_{\suc\in\PP} c((s,u_\suc))$. Hence the \OB problem for the $k$-approval voting rule for the maximum displacement distance reduces to the maximum flow problem which proves the result.
\end{proof}
}

\begin{theorem}\label{thm:bucklin_poly}
 The \ODB problem is polynomial time solvable for the simplified Bucklin voting rule for the swap distance distance and the maximum displacement distance when $\delta_i=1$ for every preference $i$ and thus for the footrule distance when $\delta_i\le 3$ for every preference $i$.
\end{theorem}

\begin{proof}
 The proof of the statement uses arguments analogous to the proof of \Cref{thm:kapp_poly}.
\end{proof}

\begin{theorem}\label{thm:bucklin_maxdis_poly}
 The \OB problem is polynomial time solvable for the simplified Bucklin voting rule for the maximum displacement distance.
\end{theorem}

\begin{proof}
  Let $\delta$ be the input distance, \PP the input profile, and $x$ the distinguished alternative. For the simplified Bucklin voting rule, we first move the distinguished alternative $x$ to its left by $\min\{\delta, pos(x, \suc)\}$ positions in every preference $\suc\in\PP$. Let the resulting profile be \QQ. Let $k$ be the minimum positive integer such that the alternative $x$ appears within the first $k$ positions in a majority number of preferences in \QQ. To make $x$ the unique simplified Bucklin winner, every alternative other than $x$ should not appear within the first $k$ positions in a majority number of preferences. This can be achieved by reducing to the maximum flow problem similarly to the proof of \Cref{thm:kapp_maxdis_poly}.
\end{proof}

\section{Hardness Results}\label{sec:hard}

In this section, we present our hardness results. We use the following\longversion{ \TSAT} problem to prove our hardness results which is known to be \NPC~\longversion{\cite{ECCC-TR03-022}}\shortversion{\shortcite{ECCC-TR03-022}}.

\begin{definition}[\TSAT]
 Given a set $\XX=\{x_i: i\in[n]\}$ of $n$ variables and a set $\{C_j: j\in[m]\}$ of $m$ $3$-CNF clauses on \XX such that, for every $i\in[n]$, $x_i$ and $\bar{x}_i$ each appear in exactly $2$ clauses, compute if there exists any Boolean assignment to the variables which satisfy all the $m$ clauses simultaneously.
\end{definition}

The high level idea of all our proofs is as follows. We have two set of preferences (let us call them ``primary'' gadgets): the first set corresponds to the variables of the \TSAT formula which ensures that every variable is assigned either true or false but not both; the second set corresponds to the clauses which ensures that every clause is satisfied by the assignment set by the first set of preferences. To ensure that our primary gadgets function in the intended manner, we add ``auxiliary'' preferences which varies for different voting rules. We now present our hardness result for the $k$-approval voting rule for the swap and footrule distances. In this section, \delta always denotes the distance under consideration.

\begin{theorem}\label{thm:ob_kapp_swap}
 \longversion{The \OB problem is \NPC for the $k$-approval voting rule for any constant $k\ge 2$ for the swap distance even when $\delta=2$. Hence, the \OB problem is \NPC for the $k$-approval voting rule for any constant $k\ge 2$ for the footrule distance even when $\delta=4$.}
 \shortversion{The \OB problem is \NPC for the $k$-approval voting rule for any constant $k\ge 2$ for the swap distance even when $\delta=2$ and for the footrule distance even when $\delta=4$.}
\end{theorem}

\begin{proof}
 Let us present the proof for the $2$--approval voting rule first. The \OB problem for the $2$-approval voting rule for the swap distance is clearly in \NP. To prove \NP-hardness, we reduce from \TSAT to \OB. Let $(\XX=\{x_i: i\in[n]\},\CC=\{C_j: j\in[m]\})$ be an arbitrary instance of \TSAT. Let us assume without loss of generality that both $n$ and $m$ are even integers; if not, we take disjoint union of the instance with itself which doubles both the number of variables and the number of clauses. Let us consider the following instance $(\AA,\PP,c,\delta=2)$ of \OB.
 \begin{align*}
  \AA &= \{a(x_i,0), a(x_i,1), a(\bar{x}_i,0), a(\bar{x}_i,1): i\in[n]\} \cup \{c,u\}\\
  &\cup \{w_i, z_i: i\in[n]\} \cup \{y_j, d_j, d^\pr_j: j\in[m]\}
 \end{align*}
 
 We construct the profile \PP using the following function $f$. The function $f$ takes a literal and a clause as input, and outputs a value in $\{0,1,-\}$. For each literal $l$, let $C_i$ and $C_j$ with $1\le i<j\le m$ be the two clauses where $l$ appears. We define $f(l,C_i)=0, f(l,C_j)=1,$ and $f(l,C_k)=-$ for every $k\in[m]\setminus\{i,j\}$. \longversion{This finishes the description of the function $f$.} We now describe \PP. Let $\Delta=100m^2n^2$.
 \begin{enumerate}[(I)]
  \item For every $i\in[n]$, we have\label{thm:ob_kapp_1}
  \begin{itemize}
   \item $w_i\suc a(x_i,0)\suc a(x_i,1)\suc z_i\suc \text{others}$
   \item $w_i\suc a(\bar{x}_i,0)\suc a(\bar{x}_i,1)\suc z_i\suc \text{others}$
  \end{itemize}
  \item For every $C_j=(l_1\vee l_2\vee l_3)$, $j\in[m]$, we have\label{thm:ob_kapp_2}
  \begin{itemize}
   \item $y_j\suc d_j\suc a(l_1,f(l_1,C_j))\suc d_j^\pr\suc \text{others}$
   \item $y_j\suc d_j\suc a(l_2,f(l_2,C_j))\suc d_j^\pr\suc \text{others}$
   \item $y_j\suc d_j\suc a(l_3,f(l_3,C_j))\suc d_j^\pr\suc \text{others}$
  \end{itemize}
  \item $c\suc d_1\suc d_2\suc d_3\suc \text{others}$\label{thm:ob_kapp_3}
  \item $\Delta+2$ copies: $u\suc d_1\suc d_2\suc c\suc \text{others}$\label{thm:ob_kapp_4}
  \item For every $1\le i\le\nfrac{n}{2}$, we have\label{thm:ob_kapp_5}
  \begin{itemize}
   \item $\Delta+1$ copies: $u\suc w_{2i-1}\suc w_{2i}\suc d_1\suc \text{others}$
  \end{itemize}
  \item For every $i\in[n]$\label{thm:ob_kapp_6}
  \begin{itemize}
   \item $\Delta+1$ copies: $u\suc a(x_i,1)\suc a(\bar{x}_i,1)\suc d_1\suc \text{others}$
   \item $\Delta+1$ copies: $u\suc a(x_i,0)\suc a(\bar{x}_i,0)\suc d_1\suc \text{others}$
  \end{itemize}
  \item For every $1\le i\le\nfrac{m}{2}$, we have\label{thm:ob_kapp_7}
  \begin{itemize}
   \item $\Delta$ copies: $u\suc y_{2j-1}\suc y_{2j}\suc d_1\suc \text{others}$
  \end{itemize}
 \end{enumerate}
 We claim that the two instances are equivalent. In one direction, let the \TSAT instance be a \YES instance with a satisfying assignment $g:\XX\longrightarrow\{0,1\}$. Let us consider the following profile \QQ where the swap distance of every preference is at most $2$ from its corresponding preference in \PP.
 \begin{enumerate}[(I)]
  \item For every $i\in[n]$, we have
  \begin{itemize}
   \item If $g(x_i)=1$, then
   \begin{itemize}
    \item $w_i\suc z_i\suc a(x_i,0)\suc a(x_i,1)\suc \text{others}$
    \item $a(\bar{x}_i,0)\suc a(\bar{x}_i,1)\suc w_i\suc z_i\suc \text{others}$
   \end{itemize}
   \item Else
   \begin{itemize}
    \item $a(x_i,0)\suc a(x_i,1)\suc w_i\suc z_i\suc \text{others}$
    \item $w_i\suc z_i\suc a(\bar{x}_i,0)\suc a(\bar{x}_i,1)\suc \text{others}$
   \end{itemize}
  \end{itemize}
  \item For every $j\in[m]$, if $C_j=(l_1\vee l_2\vee l_3)$ and $g$ makes $l_1=1$ (we can assume by renaming), then we have
  \begin{itemize}
   \item $d_j\suc a(l_1,f(l_1,C_j))\suc y_j\suc d_j^\pr\suc \text{others}$
   \item $y_j\suc d_j\suc a(l_2,f(l_2,C_j))\suc d_j^\pr\suc \text{others}$
   \item $y_j\suc d_j\suc a(l_3,f(l_3,C_j))\suc d_j^\pr\suc \text{others}$
  \end{itemize}
  \item $c\suc d_1\suc d_2\suc d_3\suc \text{others}$
  \item $\Delta+2$ copies: $u\suc c\suc d_1\suc d_2\suc \text{others}$
  \item For every $1\le i\le\nfrac{n}{2}$, we have
  \begin{itemize}
   \item $\Delta+1$ copies: $w_{2i-1}\suc w_{2i}\suc u\suc d_1\suc \text{others}$
  \end{itemize}
  \item For every $i\in[n]$
  \begin{itemize}
   \item $\Delta+1$ copies: $a(x_i,1)\suc a(\bar{x}_i,1)\suc u\suc d_1\suc \text{others}$
   \item $\Delta+1$ copies: $a(x_i,0)\suc a(\bar{x}_i,0)\suc u\suc d_1\suc \text{others}$
  \end{itemize}
  \item For every $1\le i\le\nfrac{m}{2}$, we have
  \begin{itemize}
   \item $\Delta$ copies: $y_{2j-1}\suc y_{2j}\suc u\suc d_1\suc \text{others}$
  \end{itemize}
 \end{enumerate}
 
 The $2$-approval score of every alternative from \QQ is given in \Cref{tbl:score_summary_app}. Hence, $c$ wins uniquely in \QQ and thus the \OB instance is a \YES instance.
 \begin{table}[!htbp]
  \begin{center}
   \begin{tabular}{|cc|}\hline
   Alternatives & Score\\\hline
   $c$& $\Delta+3$ \\\hline
   $w_i, i\in[n], u, y_j, j\in[m]$& $\Delta+2$ \\\hline
   \makecell{$a(x_i,0), a(x_i,1), a(\bar{x}_i,0), a(\bar{x}_i,1), i\in[n]$}& $\le\Delta+2$ \\\hline
   \makecell{$z_i, i\in[n], d_j, d^\pr_j, j\in[m]$}& $<\Delta$ \\\hline   
   \end{tabular}
  \end{center}
  \caption{Summary of scores from \QQ}\label{tbl:score_summary_app}
 \end{table}

 In the other direction, let us assume that there exists a profile \QQ such that the swap distance of every preference in \QQ is at most $2$ from its corresponding preference in \PP and $c$ wins uniquely in \QQ. We observe that $c$ does not receive any score from any preference in \Cref{thm:on_kapp_4} in \PP. As the $2$-approval winner is the alternative having the highest $2$-approval score, we can assume without loss of generality that $c$ appears within the first $2$ positions in every preference in \QQ which corresponds to the preferences in \Cref{thm:ob_kapp_3,thm:ob_kapp_4} in \PP. We  now observe that, irrespective of other preferences in \QQ, the score of $c$ in \QQ is $\Delta+3$. Also the score of $u$ from the preferences corresponding to \Cref{thm:ob_kapp_4} in \QQ is already $\Delta+2$. Hence, for $c$ to win uniquely in \QQ, $u$ must not appear within the first two positions in any preference in \QQ which corresponds to the preferences in \Cref{thm:ob_kapp_5,thm:ob_kapp_6,thm:ob_kapp_7} in \PP. This makes the score of $w_i, i\in[n]$ $\Delta+1$, the score of $a(l,0)$ $\Delta+2$ and the score of $a(l,1)$ $\Delta+1$ for every literal $l$, and the score of $y_j, j\in[m]$ $\Delta$ from the preferences corresponding to  \Cref{thm:ob_kapp_3,thm:ob_kapp_4,thm:ob_kapp_5,thm:ob_kapp_6,thm:ob_kapp_7} in \QQ. We now consider the following assignment $g:\XX\longrightarrow\{0,1\}$ -- for every $i\in[n]$, $g(x_i)=1$ if the preference corresponding to $w_i\suc a(x_i,0)\suc a(x_i,1)\suc z_i\suc \text{others}$ in \QQ keeps $w_i$ within the first $2$ positions; if the preference corresponding to $w_i\suc a(\bar{x}_i,0)\suc a(\bar{x}_i,1)\suc z_i\suc \text{others}$ in \QQ keeps $w_i$ within the first $2$ positions, then we define $g(x_i)=0$. We observe that the function $g$ is well defined since $w_i$ must be pushed out of the first $2$ positions at least once in the preferences corresponding to $w_i\suc a(x_i,0)\suc a(x_i,1)\suc z_i\suc \text{others}$ and $w_i\suc a(\bar{x}_i,0)\suc a(\bar{x}_i,1)\suc z_i\suc \text{others}$; otherwise the score of $w_i$ is the same as the score of $c$ which is a contradiction. We claim that $g$ is a satisfying assignment for the \TSAT instance. Suppose not, then let us assume that $g$ does not satisfy $C_j=(l_1\vee l_2\vee l_3)$ for some $j\in[m]$. Since the literals $l_1, l_2,$ and $l_3$ are set to $0$ by $g$, it follows from the definition of $g$ that the score of every alternative in $\{a(l_i,\mu):i\in[3],\mu\in\{0,1\}\}$ (which is a superset of $\{a(l_i,f(l_i,C_j)):i\in[3]\}$) from the preferences corresponding to \Cref{thm:ob_kapp_1,thm:ob_kapp_3,thm:ob_kapp_4,thm:ob_kapp_5,thm:ob_kapp_6,thm:ob_kapp_7} in \QQ is already $\Delta+2$ (since all these alternatives appear within the first $2$ positions in one preference in \Cref{thm:ob_kapp_1} in \QQ). Hence $y_j$ can never be pushed out of the first $2$ positions in the preferences corresponding to \Cref{thm:ob_kapp_2} in \QQ. However this makes the score of $y_j$ $\Delta+3$ in \QQ. This contradicts our assumption that $c$ is the unique winner in \QQ. Hence $g$ is a satisfying assignment and thus the \TSAT instance is a \YES instance.
 
 Generalizations to $k$--approval for any constant $k\ge 3$ can be done by using $10m^8n^8$ dummy candidates and ensuring that any dummy alternative appears at most once within the first $k+3$ positions in the profile \PP. The proof for the footrule distance follows from the observation that, for any two preferences $\suc_1, \suc_2\in\LL(\AA)$, if we have $d_{\text{footrule}}(\suc_1,\suc_2)=4$, then we have $d_{\text{swap}}(\suc_1,\suc_2)=2$.
\end{proof}

\begin{theorem}\label{thm:kapp_max_hard}
 The \ODB problem is \NPC for the $k$-approval voting rule for the maximum displacement distances even when $\delta_i=2$ for every preference $i$ for any constant $k>1$.
\end{theorem}

\begin{proof}
 The \ODB problem for the $k$-approval voting rule for the maximum displacement distance is clearly in \NP. To prove \NP-hardness, we reduce from \TSAT to \ODB. Let $(\XX=\{x_i: i\in[n]\},\CC=\{C_j: j\in[m]\})$ be an arbitrary instance of \TSAT. Let us consider the following instance of \ODB. The set of alternatives \AA is 
 $$\AA = \{a_i, \bar{a}_i, b_i, \bar{b}_i, w_i, w_i^\pr: i\in[n]\} \cup \{c\} \cup \{y_j: j\in[m]\}\cup \DD, \text{ where } |\DD|=300m^3 n^3$$ 
 
 We now describe the preference profile \PP along with their corresponding prices. The profile \PP is a disjoint union of two profiles, namely, $\PP_1$ and $\PP_2$. While describing the preferences of \PP below, whenever we say `others' or `for some alternatives in \DD' or `for some subset of \DD', the unspecified alternatives are assumed to be arranged in such a way that, for every unspecified alternative $a\in\AA\setminus\DD$, at least $10$ alternatives from \DD appear immediately before $a$. We also ensure that any alternative in \DD appears within top $k+10$ positions at most once in $\PP$ whereas every alternative in $\AA\setminus\DD$ appears within top $10mn$ position in every preference in $\PP$. This is always possible because $|\DD|=300m^3 n^3$ whereas $|\AA\setminus\DD|=6n+m+1$ and $|\PP|\le 100(m+n)$. Let $f$ and $g$ be functions defined on the set of literals as $f(x_i)=a_i, g(x_i)=b_i,$ and $f(\bar{x}_i)=\bar{a}_i, g(\bar{x}_i)=\bar{b}_i$ for every $i\in[n]$. We also define a function $h$ as $h(C_j, l)=f(l)$ if the literal $l$ appears in the clause $C_j$ but $l$ does not appear in any clause $C_r$ for any $1\le r<l$, $h(C_j, l)=g(l)$ if the literal $l$ appears in the clause $C_j$ but $l$ does not appear in any clause $C_r$ for any $l< r\le m$, and $-$ otherwise. We first describe $\PP_1$ below. The price of every preference in $\PP_1$ is $1$. 
 
 \begin{enumerate}[(I)]
  \item For every $i\in[n]$, we have the following preferences.
  \begin{itemize}\label{kapp:votesi}
   \item $\DD_{k-2}\suc w_i\suc w_i^\pr\suc a_i\suc b_i\suc \text{others}$ for some $\DD_{k-2}\subseteq \DD$ with $|\DD_{k-2}|=k-2$
   \item $\DD_{k-2}\suc w_i\suc w_i^\pr\suc \bar{a}_i\suc \bar{b}_i\suc \text{others}$ for some $\DD_{k-2}\subseteq \DD$ with $|\DD_{k-2}|=k-2$
  \end{itemize}
  \item For every $C_j=(l_1\vee l_2\vee l_3)$ with $j\in[m]$, we have the following preferences.
  \begin{itemize}\label{kapp:votesj}
   \item $\DD_{k-2}\suc d\suc y_j\suc h(C_j,l_1)\suc d^\pr\suc \text{others}$, for some $d,d^\pr\in\DD$ for some $\DD_{k-2}\subseteq \DD$ with $|\DD_{k-2}|=k-2$
   \item $\DD_{k-2}\suc d\suc y_j\suc h(C_j,l_2)\suc d^\pr\suc \text{others}$, for some $d,d^\pr\in\DD$ for some $\DD_{k-2}\subseteq \DD$ with $|\DD_{k-2}|=k-2$
   \item $\DD_{k-2}\suc d\suc y_j\suc h(C_j,l_3)\suc d^\pr\suc \text{others}$, for some $d,d^\pr\in\DD$ for some $\DD_{k-2}\subseteq \DD$ with $|\DD_{k-2}|=k-2$
  \end{itemize}
 \end{enumerate}

 We now describe the preferences in $\PP_2$. The price of every preference in $\PP_2$ is $10mn$.
 \begin{enumerate}[(I)]
  \item We have $10$ copies of $\DD_{k-2}\suc c\suc d\suc \text{others}$, for some $d\in\DD, \DD_{k-2}\subseteq \DD$ with $|\DD_{k-2}|=k-2$
  \item For every $x\in\{a_i, \bar{a}_i, b_i, \bar{b}_i, w_i, w_i^\pr: i\in[n]\}$, we have $8$ copies of the following preference.
  \begin{itemize}
   \item $\DD_{k-2}\suc x\suc d\suc \text{others}$, for some $d\in\DD, \DD_{k-2}\subseteq \DD$ with $|\DD_{k-2}|=k-2$
  \end{itemize}
  \item For every $x\in\{y_j: j\in[m]\}$, we have $7$ copies of the following preference.
  \begin{itemize}
   \item $\DD_{k-2}\suc x\suc d\suc \text{others}$, for some $d\in\DD, \DD_{k-2}\subseteq \DD$ with $|\DD_{k-2}|=k-2$
  \end{itemize}
 \end{enumerate}
 
 The budget of the briber is $(m+n)$. The value of $\delta_i$ for every voter $i$ is $2$. This finishes the description of the \TSAT instance. We now claim that the two instances are equivalent. In one direction, let the \TSAT instance be a \YES instance with a satisfying assignment $\gamma:\{x_i:i\in[n]\}\longrightarrow\{0,1\}$. Let us consider the following profile \QQ where the maximum displacement distance of every preference is at most $2$ from its corresponding preference in \PP.
 \begin{enumerate}[(I)]
  \item For every $i\in[n]$, if $\gamma(x_i)=1$, then we have the following preferences.
  \begin{itemize}
   \item $w_i\suc w_i^\pr\suc a_i\suc b_i\suc \text{others}$
   \item $\bar{a}_i\suc \bar{b}_i\suc w_i\suc w_i^\pr\suc \text{others}$
  \end{itemize}
  \item For every $i\in[n]$, if $\gamma(x_i)=0$, then we have the following preferences.
  \begin{itemize}\label{kapp:votesi}
   \item $a_i\suc b_i\suc w_i\suc w_i^\pr\suc \text{others}$
   \item $w_i\suc w_i^\pr\suc \bar{a}_i\suc \bar{b}_i\suc \text{others}$
  \end{itemize}
  \item For every $C_j=(l_1\vee l_2\vee l_3)$ with $j\in[m]$ and $\gamma$ makes $l_1=1$ (we can assume by renaming), then we have the following preferences.
  \begin{itemize}
   \item $d\suc h(C_j,l_1)\suc y_j\suc d^\pr\suc \text{others}$, for some $d,d^\pr\in\DD$
   \item $d\suc y_j\suc h(C_j,l_2)\suc d^\pr\suc \text{others}$, for some $d,d^\pr\in\DD$
   \item $d\suc y_j\suc h(C_j,l_3)\suc d^\pr\suc \text{others}$, for some $d,d^\pr\in\DD$
  \end{itemize}
 \end{enumerate}
 
 The preferences in $\PP_2$ remain same in \QQ. Also since exactly $(n+m)$ preferences in $\PP_1$ change in \QQ, the total prices of the preferences which are changed is \BB. The $k$-approval score of every alternative in the preference profile \QQ is shown in \Cref{tbl:score_summary_kapp}. Hence the \ODB instance is a \YES instance.
 \begin{table}[!htbp]
 	\begin{center}
 		\begin{tabular}{|cc|}\hline
 			Alternatives & Score\\\hline
 			$c$& $10$ \\\hline
 			$\AA\setminus(\DD\cup\{c\})$& $9$ \\\hline
 			\DD& $\le 1$ \\\hline   
 		\end{tabular}
 	\end{center}
 	\caption{Summary of scores from \QQ}\label{tbl:score_summary_kapp}
 \end{table}

 In the other direction, let us assume that there exists a profile \QQ such that the maximum displacement distance of every preference in \QQ is at most $2$ from its corresponding preference in \PP, $c$ is the unique $k$-approval winner in \QQ, and the sum of the prices of the preferences which differ in \PP and \QQ is at most $(m+n)$. We first observe that the score of $c$ in \QQ is the same as in \PP which is $10$. Also since every alternative in \DD appears within the first $k+10$ positions at most once in \PP, no alternative in \DD can win in \QQ. Also since the budget is $(m+n)$ and the price of every preference in $\PP_2$ is more than $(m+n)$, no preference in $\PP_2$ change in \QQ. Let the preferences in \QQ which correspond to $\PP_1$ be $\QQ_1$. Hence, for $c$ to win uniquely in \QQ, every alternative in $\{a_i, \bar{a}_i, b_i, \bar{b}_i, w_i, w_i^\pr: i\in[n]\}$ can appear within the first $k$ positions at most once in $\QQ_1$. Also every alternative in $\{y_j: j\in[m]\}$ can appear within the first $k$ positions at most twice in $\QQ_1$ for $c$ to win uniquely. Since the budget is $(m+n)$ and the price of every preference in $\PP_1$ is $1$, for every $i\in[n]$, there will be exactly one preference in $\QQ_1$ where both $w_i$ and $w_i^\pr$ appear within the first $k$ positions and there will be exactly two preferences in $\QQ_1$ where $y_j$ appears within the first $k$ positions for every $j\in[m]$. Let us now consider the following assignment $\gamma:\{x_i: i\in[n]\}\longrightarrow\{0,1\}$ defined as: $\gamma(x_i)=0$ if the preference in $\QQ_1$ where both $w_i$ and $w_i^\pr$ appear within the first $k$ positions has $a_i$ at the $(k+1)$-th position; otherwise $\gamma(x_i)=1$. We claim that $\gamma$ is a satisfying assignment for the \TSAT instance. Suppose not, then let us assume that $\gamma$ does not satisfy $C_j=(l_1\vee l_2\vee l_3)$ for some $j\in[m]$. We have already observed that, for $c$ to become the unique $k$-approval winner in \QQ, the alternative $y_j$ must move to its right by one position in at least one of the $3$ preferences in $\PP_1$ where it appears at the first position. However, it follows from the definition of $\gamma$ that this would make the $k$-approval score of at least one alternative in $\{f(l_1), f(l_2), f(l_3)\}$ same as the $k$-approval score of $c$ which contradicts our assumption that $c$ is the unique winner in \QQ. Hence $\gamma$ is a satisfying assignment and thus the \TSAT instance is a \YES instance.
 
\end{proof}

We next present our result for the Borda voting rule.

\begin{theorem}\label{thm:borda_max}
 \longversion{The \OB problem is \NPC for the Borda voting rule for the swap and the maximum displacement distances even when $\delta=1$. Hence, the \OB problem is \NPC for the Borda voting rule for the footrule distance even when $\delta=2$.}
 \shortversion{The \OB problem is \NPC for the Borda voting rule for the swap and the maximum displacement distances even when $\delta=1$ and thus for the footrule distance even when $\delta=2$.}
\end{theorem}

\begin{proof}
 Let us prove the result for the maximum displacement distance first. The \OB problem for the Borda voting rule for the maximum displacement distance is clearly in \NP. To prove \NP-hardness, we reduce from \TSAT to \OB. Let $(\XX=\{x_i: i\in[n]\},\CC=\{C_j: j\in[m]\})$ be an arbitrary instance of \TSAT. Let us consider the following instance $(\AA,\PP,c,\delta=1)$ of \OB.
 \begin{align*}
  \AA &= \{a_i, \bar{a}_i, z_i: i\in[n]\} \\
  &\cup \{c\} \cup \{y_j: j\in[m]\}\\
  &\cup \DD, \text{ where } |\DD|=10m^7 n^7
 \end{align*}
 
 We construct the profile \PP which is a disjoint union of two profiles, namely, $\PP_1$ and $\PP_2$. We first describe $\PP_1$ below. While describing the preferences below, whenever we say `others' or `for some alternative in \DD' or `for some subset of \DD', the unspecified alternatives are assumed to be arranged in such a way that, for every unspecified alternative $a\in\AA\setminus\DD$, at least $10$ alternatives from \DD appear immediately before $a$. We also ensure that any alternative in \DD appears within top $10m^2n^2$ positions at most once in $\PP_1$ whereas every alternative in $\AA\setminus\DD$ appears within top $10mn$ position in every preference in $\PP_1$. This is always possible because $|\DD|=10m^7 n^7$ whereas $|\AA\setminus\DD|=3n+m+1$ and $|\PP_1|=2n+3m$. Let $f$ be a function defined on the set of literals as $f(x_i)=a_i$ and $f(\bar{x}_i)=\bar{a}_i$ for every $i\in[n]$. 
 
 \begin{enumerate}[(I)]
  \item For every $i\in[n]$, we have the following preferences.
  \begin{itemize}\label{borda:votesi}
   \item $z_i\suc a_i\suc d\suc d^\pr\suc c\suc \text{others}$, for some $d, d^\pr\in\DD$
   \item $z_i\suc \bar{a}_i\suc d\suc d^\pr\suc c\suc \text{others}$, for some $d, d^\pr\in\DD$
  \end{itemize}
  \item For every $C_j=(l_1\vee l_2\vee l_3)$ with $j\in[m]$, we have the following preferences.
  \begin{itemize}\label{borda:votesj}
   \item $y_j\suc f(l_1)\suc d\suc c\suc\text{others}$, for some $d\in\DD$
   \item $y_j\suc f(l_2)\suc d\suc c\suc\text{others}$, for some $d\in\DD$
   \item $y_j\suc f(l_3)\suc d\suc c\suc\text{others}$, for some $d\in\DD$
  \end{itemize}
 \end{enumerate}
 
 This finishes the description of $\PP_1$. Let the function $s_1(\cdot)$ maps every alternative in \AA to the Borda score it receives from $\PP_1$. We observe that $s_1$ is integer valued and $s_1(c)$ is the unique maximum of $s_1$ for large enough $n$ and $m$. We now describe the profile $\PP_2$. Let $N_1$ be the number of preferences in $\PP_1$. While describing the preferences in $\PP_2$, we will leave the choice of alternatives in \DD unspecified. However, we assume that those choice are made in such a way that any alternative in \DD appears within the first $10mn$ positions at most once in $\PP_2$. This is always possible since $|\DD|=10m^7 n^7$ and $|\PP_2|\le 10 m^3n^3$. Let $M=|\AA\setminus\DD|$; we have $M=4n+m+1$.
 
 \begin{enumerate}[(I)]
  \item For every $i\in[n]$, we have the following preferences.
  \begin{itemize}
   \item $s_1(c)-s_1(z_i)+2N_1-2$ preferences of type $d_0\suc d_1\suc c\suc b_1\suc d_2\suc b_2\suc d_3\suc \cdots\suc b_{M-2}\suc d_{M-1}\suc z_i\suc \text{others}$, where $\AA\setminus\DD=\{c,z_i\}\cup\{b_k: k\in[M-2]\}$ and $d_\el\in\DD$ for $0\le \el\le M-1$.
   \item $s_1(c)-s_1(z_i)+2N_1-2$ preferences of type $z_i\suc d_0\suc d_1\suc b_{M-2}\suc d_2\suc \cdots\suc b_1\suc d_{M-1}\suc d_M\suc d_{M+1}\suc c\suc\text{others}$, where $\AA\setminus\DD=\{c,z_i\}\cup\{b_k: k\in[M-2]\}$ and $d_\el\in\DD$ for $0\le \el\le M+1$.
   
   \item $s_1(c)-s_1(a_i)+2N_1-5$ preferences of type $d_0\suc d_1\suc c\suc b_1\suc d_2\suc b_2\suc d_3\suc \cdots\suc b_{M-2}\suc d_{M-1}\suc a_i\suc \text{others}$, where $\AA\setminus\DD=\{c,a_i\}\cup\{b_k: k\in[M-2]\}$ and $d_\el\in\DD$ for $0\le \el\le M-1$.
   \item $s_1(c)-s_1(a_i)+2N_1-5$ preferences of type $a_i\suc d_0\suc d_1\suc b_{M-2}\suc d_2\suc \cdots\suc b_1\suc d_{M-1}\suc d_M\suc d_{M+1}\suc c\suc\text{others}$, where $\AA\setminus\DD=\{c,a_i\}\cup\{b_k: k\in[M-2]\}$ and $d_\el\in\DD$ for $0\le \el\le M+1$.
   
   \item $s_1(c)-s_1(\bar{a}_i)+2N_1-5$ preferences of type $d_0\suc d_1\suc c\suc b_1\suc d_2\suc b_2\suc d_3\suc \cdots\suc b_{M-2}\suc d_{M-1}\suc \bar{a}_i\suc \text{others}$, where $\AA\setminus\DD=\{c,\bar{a}_i\}\cup\{b_k: k\in[M-2]\}$ and $d_\el\in\DD$ for $0\le \el\le M-1$.
   \item $s_1(c)-s_1(\bar{a}_i)+2N_1-5$ preferences of type $\bar{a}_i\suc d_0\suc d_1\suc b_{M-2}\suc d_2\suc \cdots\suc b_1\suc d_{M-1}\suc d_M\suc d_{M+1}\suc c\suc\text{others}$, where $\AA\setminus\DD=\{c,\bar{a}_i\}\cup\{b_k: k\in[M-2]\}$ and $d_\el\in\DD$ for $0\le \el\le M+1$.
  \end{itemize}
  
  \item For every $j\in[m]$, we have the following preferences.
  \begin{itemize}
   \item $s_1(c)-s_1(y_j)+2N_1-2$ preferences of type $d_0\suc d_1\suc c\suc b_1\suc d_2\suc b_2\suc d_3\suc \cdots\suc b_{M-2}\suc d_{M-1}\suc y_j\suc \text{others}$, where $\AA\setminus\DD=\{c,y_j\}\cup\{b_k: k\in[M-2]\}$ and $d_\el\in\DD$ for $0\le \el\le M-1$.
   \item $s_1(c)-s_1(y_j)+2N_1-2$ preferences of type $y_j\suc d_0\suc d_1\suc b_{M-2}\suc d_2\suc \cdots\suc b_1\suc d_{M-1}\suc d_M\suc d_{M+1}\suc c\suc\text{others}$, where $\AA\setminus\DD=\{c,y_j\}\cup\{b_k: k\in[M-2]\}$ and $d_\el\in\DD$ for $0\le \el\le M+1$.
  \end{itemize}
 \end{enumerate}

 This finishes the description of $\PP_2$ and thus of \PP. Since $s_1(c)$ is upper bounded by $10m^7n^7$, \PP has only polynomially in ($m,n$) many preferences. Let $N_2=|\PP_2|$ and $N=|\PP|=N_1+N_2$. We summarize the Borda scores of every alternative from \PP in \Cref{tbl:borda_md}. Let the function $s(\cdot)$ maps every alternative in \AA to the Borda score it receives from $\PP$. We now claim that the two instance are equivalent.
 
 \begin{table*}[!htbp]
  \centering
  \begin{tabular}{c|c}
   Alternatives & Borda scores from \PP\\\hline\hline
   
   $a_i, \bar{a}_i \forall i\in[n]$ & $s(c)+2N-5$ \\
   $z_i, \forall i\in[n]$ & $s(c)+2N-2$\\
   $y_j, \forall j\in[m]$ & $s(c)+2N-2$\\
   $d\in\DD$ & $<s(c)-10m^2n^2$\\\hline
  \end{tabular}
  \caption{Borda scores of the alternatives in \AA from \PP.}\label{tbl:borda_md}
 \end{table*}
 
 In one direction, let us assume that the \TSAT instance is a \YES instance. Let $g:\XX\longrightarrow\{0,1\}$ be a satisfying assignment for the \TSAT instance. Let us consider the following profile \QQ where the maximum displacement distance of every preference in \QQ from its corresponding preference in \PP is at most $1$. The preferences in \QQ which corresponds to $\PP_1$ are as follows.
 
 \begin{enumerate}[(I)]
  \item For every $i\in[n]$, we have the following preferences.
  \begin{itemize}
   \item If $g(x_i)=1$, then
  \begin{itemize}
   \item $z_i\suc d\suc a_i\suc c\suc d^\pr\suc \text{others}$, for some $d, d^\pr\in\DD$
   \item $\bar{a}_i\suc z_i\suc d\suc c\suc d^\pr\suc \text{others}$, for some $d, d^\pr\in\DD$
  \end{itemize}
  \item Else
  \begin{itemize}
   \item $a_i\suc z_i\suc d\suc c\suc d^\pr\suc \text{others}$, for some $d, d^\pr\in\DD$
   \item $z_i\suc d\suc \bar{a}_i\suc c\suc d^\pr\suc \text{others}$, for some $d, d^\pr\in\DD$
  \end{itemize}
  \end{itemize}
  \item For every $j\in[m]$, if $C_j=(l_1\vee l_2\vee l_3)$ and $g$ makes $l_1=1$ (we can assume by renaming), then we have
  \begin{itemize}
   \item $f(l_1)\suc y_j\suc c\suc d\suc\text{others}$, for some $d\in\DD$
   \item $y_j\suc f(l_2)\suc c\suc d\suc\text{others}$, for some $d\in\DD$
   \item $y_j\suc f(l_3)\suc c\suc d\suc\text{others}$, for some $d\in\DD$
  \end{itemize}
 \end{enumerate}
 
 In every preference in \QQ corresponding to $\PP_2$, the alternative $c$ is moved to its left by $1$ position; it is possible because $c$ does not appear at the first position in any preference in $\PP_2$. We also move every alternative in $\AA\setminus(\DD\cup\{c\})$ to its right by $1$ position; it is possible since no alternative in $\AA\setminus(\DD\cup\{c\})$ appears at the last position in any preference in $\PP_2$ and every alternative in $\AA\setminus(\DD\cup\{c\})$ is followed immediately by some alternative in \DD in every preference in $\PP_2$. We summarize the Borda scores of every alternative from \QQ in \Cref{tbl:borda_1_fwd}. Hence $c$ is the unique Borda winner in \QQ and thus the \OB instance is a \YES instance.
 
 \begin{table*}[!htbp]
  \centering
  \begin{tabular}{c|c}
   Alternatives & Borda scores from \QQ\\\hline\hline
   
   $a_i, \bar{a}_i \forall i\in[n]$ & $s(c)-1$ \\
   $z_i, \forall i\in[n]$ & $s(c)-1$\\
   $y_j, \forall j\in[m]$ & $s(c)-1$\\
   $d\in\DD$ & $<s(c)-10mn$\\\hline
  \end{tabular}
  \caption{Borda scores of the alternatives in \AA from \QQ in the proof of \Cref{thm:borda_max}.}\label{tbl:borda_1_fwd}
 \end{table*}
 
 In the other direction, let us assume that there exists a profile \QQ such that the maximum displacement distance of every preference in \QQ is at most $1$ from its corresponding preference in \PP and $c$ is the unique Borda winner in \QQ. We first observe that, irrespective of \QQ (subject to the condition that its maximum displacement distance from \PP is at most $1$), no alternative from \DD wins in \QQ under the Borda voting rule. We can assume without loss of generality that, in every preference in \QQ, the alternative $c$ is moved to its left by $1$ position since $c$ never appears at the first position in any preference in \PP. We can also assume without loss of generality that every alternative in $\AA\setminus(\DD\cup\{c\})$ moves to its right by $1$ position in every preference corresponding to $\PP_2$ since no alternative in $\AA\setminus(\DD\cup\{c\})$ appears at the last position in any preference in $\PP_2$ and every alternative in $\AA\setminus(\DD\cup\{c\})$ is followed immediately by some alternative in \DD in every preference in $\PP_2$. We now observe that, for $c$ to become the unique Borda winner in \QQ, every alternative $z_i, i\in[n]$ must move to its right by one position in at least one of the two preferences in $\PP_1$ where $z_i$ appears at the first position (which are exactly the preferences in \Cref{borda:votesi}). However, this makes either $a_i$ or $\bar{a}_i, i\in[n]$ to appear at the first position in at least one of the two preferences in $\PP_1$ where $z_i$ appears at the first position. Let us now consider the following assignment $g:\{x_i: i\in[n]\}\longrightarrow\{0,1\}$ defined as: $g(x_i)=0$ if $a_i$ appears at the first position in at least one of the two preferences in $\PP_1$ where $z_i$ appears at the first position; we define $g(x_i)=1$ otherwise. We claim that $g$ is a satisfying assignment for the \TSAT instance. Suppose not, then let us assume that $g$ does not satisfy $C_j=(l_1\vee l_2\vee l_3)$ for some $j\in[m]$. We observe that, for $c$ to become the unique Borda winner, the alternative $y_j$ must move to its right by one position in at least one of the $3$ preferences in $\PP_1$ where it appears at the first position (which are exactly the preferences in \Cref{borda:votesj}). However, it follows from the definition of $g$ that this would make the Borda score of at least one alternative in $\{f(l_1), f(l_2), f(l_3)\}$ same as the Borda score of $c$ which contradicts our assumption that $c$ is the unique Borda winner in \QQ. Hence $g$ is a satisfying assignment and thus the \TSAT instance is a \YES instance.
 
 For proving the result for the swap and footrule distances, we change the preferences in $\PP_2$ by shifting every alternative in $\AA\setminus(\DD\cup\{c\})$ to their right by $1$ position. Then analogous argument proves the result.
\end{proof}

The proof of \Cref{thm:borda_max} can be adapted to prove the following result.
\begin{theorem}\label{cor:scr}
 \longversion{Let $\alpha=(\alpha_i)_{i\in[m]}$ be a score vector with $\alpha_j-\alpha_{j+1}=\alpha_{j+1}-\alpha_{j+2}>0$ for some $j\in[\nfrac{m}{2}]$. The \OB problem is \NPC for the scoring rule with score vector $\alpha$ for the swap and the maximum displacement distances even when $\delta=1$. Hence, the \OB problem is \NPC for the scoring rule with score vector $\alpha$ for the footrule distance even when $\delta=2$.}
 \shortversion{Let $\alpha=(\alpha_i)_{i\in[m]}$ be a score vector with $\alpha_j-\alpha_{j+1}=\alpha_{j+1}-\alpha_{j+2}>0$ for some $j\in[\nfrac{m}{2}]$. The \OB problem is \NPC for the scoring rule with score vector $\alpha$ for the swap and the maximum displacement distances even when $\delta=1$ and thus for the footrule distance even when $\delta=2$.}
\end{theorem}

We will use the following structural lemma in our proofs for the maximin and Copeland$^\alpha$ voting rules.

\begin{lemma}\label{lem:maximin}
 Let $\AA=\BB\cup\CC$ be a set of alternatives, $Z_{(a,b)}, a,b\in\BB, a\ne b$, non-negative integers which are either all even or all odd, and $Z_{(a,b)}=-Z_{(b,a)}$ for every $a,b\in\BB, a\ne b$. Let $|\CC|>10 K^2|\BB|^2 \sum_{a,b\in\BB, a\ne b} |Z_{(a,b)}|$. Then there exists a profile \PP such that 
 \begin{enumerate}[(i)]
  \item $\DD_\PP(a,b)=Z_{(a,b)}$ for every $a,b\in\BB, a\ne b$, $\DD_\PP(b^\pr,c^\pr)>0$ for every $b^\pr\in\BB, c^\pr\in\CC$.
  
  \item for every alternative $b\in\BB$, there are at least $\nfrac{K}{2}$ alternatives from \CC in the immediate $\nfrac{K}{2}$ positions on both left and right of $a$.
  
  \item For every $c\in\CC$, there exists at most $1$ preference in \PP where the distance of $c$ from any alternative in \BB is less than $\nfrac{K}{2}$.
  
 \end{enumerate}
 Moreover, \PP contains poly($\sum_{a,b\in\BB, a\ne b} |Z_{(a,b)}|$) many preferences, and there is an algorithm for constructing \PP which runs in time polynomial in $\sum_{a,b\in\BB, a\ne b} |Z_{(a,b)}|+m$.
\end{lemma}

\longversion{
\begin{proof}
 Let $|\BB|=\el, \CC=\bigcup_{a,b\in\BB, a\ne b} \CC_{(a,b)}$ where $\CC_{(a,b)}$s are pairwise disjoint and $|\CC_{(a,b)}|\ge 4k\el(|Z_{(a,b)}|+|Z_{(b,z)}|)$ for every $a,b\in\BB, a\ne b$. First let us prove the result when $Z_{(a,b)}, a,b\in\BB, a\ne b$, are non-negative integers which are either all even. For $a,b\in\BB, a\ne b$, let $\CC_{(a,b)}=\bigcup_{i=1}^{2(|Z_{(a,b)}|+|Z_{(b,z)}|)} C_{(a,b)}^i$ where $C_{(a,b)}^i$s are pairwise disjoint and $|C_{(a,b)}^i|\ge 2k\el$. For $a,b\in\BB, a\ne b$ such that $Z_{(a,b)}>0$, we have the following preferences. Let $\BB=\{b_i: i\in[\el-2]\}\cup\{a,b\}$.
 \begin{itemize}
  \item $\nfrac{(Z_{(a,b)})}{2}$ copies of: $a\suc \CC_{(a,b)}^1\suc b\suc\CC_{(a,b)}^2 \suc b_1\suc \CC_{(a,b)}^3\suc \dots \suc b_{\el-2}\suc \CC_{(a,b)}^\el\cdots$
  
  \item $\nfrac{(Z_{(a,b)})}{2}$ copies of: $b_{\el-2}\suc \CC_{(a,b)}^{\el+1}\suc b_{\el-3}\suc \CC_{(a,b)}^{\el+2}\suc\cdots \suc  \CC_{(a,b)}^{2\el-1}\suc a\suc \CC_{(a,b)}^{2\el}\suc b\suc\cdots$
 \end{itemize}

 Let \PP be the resulting profile. It is immediate that \PP satisfies all the conditions in the statement. This proves the statement when $Z_{(a,b)}, a,b\in\BB, a\ne b$, are non-negative integers which are either all even.
 
 From the proof above, it follows that the statement holds even when $|\CC|=8 K^2|\BB|^2 \sum_{a,b\in\BB, a\ne b} |Z_{(a,b)}|$. Let $\BB=\{b_i: i\in[\el]\}$. Let $\CC=\CC^\pr\cup (\bigcup_{i=1}^\el \CC_i)$ where $|\CC_i|=5K$ for every $i\in[\el]$. Let us consider the following preference.
 
 $$\suc = b_1 \suc \CC_1\suc b_2\suc \CC_2 \suc \cdots\suc \CC_{\el-1}\suc b_\el\suc \CC_\el\suc\cdots$$
 
 Let us define $Z^\pr:\BB\times\BB\longrightarrow\ZB$ as follows.
 
 $$
 Z^\pr(b_i,b_j) = 
 \begin{cases}
  Z(b_i,b_j)+1 & \text{ if } i>j\\
  Z(b_i,b_j)-1 & \text{ otherwise}
 \end{cases}
 $$
 
 Let \RR be the profile which satisfies the conditions in the statement for the set of alternatives $\BB\cup\CC^\pr$ and the integers $Z_{(a,b)}^\pr, a,b\in\BB, a\ne b$. We append the alternatives in $\cup_{i=1}^\el \CC_i$ at the bottom of every preference in \RR; let the resulting profile be \QQ. Clearly, $(\QQ,\suc)$ satisfies all the conditions in the statement.
\end{proof}
}

We now present our result for the maximin voting rule.

\begin{theorem}\label{thm:ob_maximin_swap}
 \longversion{The \OB problem is \NPC for the maximin voting rule for the swap distance and the maximum displacement distance even when $\delta=1$. Hence, the \OB problem is \NPC for the maximin voting rule for the footrule distance even when $\delta=2$.}
 \shortversion{The \OB problem is \NPC for the maximin voting rule for the swap distance and the maximum displacement distance even when $\delta=1$ and thus for the footrule distance even when $\delta=2$.}
\end{theorem}

\begin{proof}
 Let us first prove the result for the swap distance. The \OB problem for the maximin voting rule for the swap distance is clearly in \NP. To prove \NP-hardness, we reduce from \TSAT to \OB. Let $(\XX=\{x_i: i\in[n]\},\CC=\{C_j: j\in[m]\})$ be an arbitrary instance of \TSAT. Let us consider the following instance $(\AA,\PP,c,\delta=1)$ of \OB.
 \begin{align*}
  \AA &= \{a_i, \bar{a}_i,w_i, z_i: i\in[n]\} \cup \{y_j: j\in[m]\}\\
  &\cup \{c, b\} \cup \DD, \text{ where } |\DD|=10m^5 n^5
 \end{align*}
 
 We construct the profile \PP which is a disjoint union of two profiles, namely, $\PP_1$ and $\PP_2$. We first describe $\PP_1$ below. While describing the preferences below, whenever we say `others' , the unspecified alternatives are assumed to be arranged in such a way that, for every unspecified alternative $a\in\AA\setminus\DD$, there are at least $10$ alternatives from \DD in the immediate $10$ positions on both left and right of $a$. We also ensure that any alternative in \DD appears within top $10mn$ positions at most once in $\PP_1$ whereas every alternative in $\AA\setminus\DD$ appears within top $10mn$ position in every preference in $\PP_1$. This is possible because $|\DD|=10m^5 n^5$ and $|\AA\setminus\DD|=4n+m+2$. Let $f$ be a function defined on the set of literals as $f(x_i)=a_i$ and $f(\bar{x}_i)=\bar{a}_i$ for every $i\in[n]$.  
 \begin{enumerate}[(I)]
  \item For every $i\in[n]$, we have the following preferences.
  \begin{itemize}
   \item $2$ copies of $z_i\suc a_i\suc w_i\suc \text{others}$
   \item $2$ copies of $z_i\suc \bar{a}_i\suc w_i\suc \text{others}$
  \end{itemize}
  \item For every $C_j=(l_1\vee l_2\vee l_3), j\in[m]$, we have the following preferences. Let $h$ be a function defined on the set of literals as $h(x_i)=w_i$ and $h(\bar{x}_i)=w_i$ for every $i\in[n]$.
  \begin{itemize}
   \item $y_j\suc f(l_1)\suc h(l_1)\suc\text{others}$
   \item $y_j\suc f(l_2)\suc h(l_2)\suc\text{others}$
   \item $y_j\suc f(l_3)\suc h(l_3)\suc\text{others}$
  \end{itemize}
 \end{enumerate}
 
 Due to \Cref{lem:maximin} (using $K=30$ in \Cref{lem:maximin}), there exists a profile $\PP_2$ consisting of $\text{poly}(m,n)$ preferences such that the weighted majority graph of the profile $\PP=\PP_1\cup\PP_2$ is as described in \Cref{tbl:maximin_score_all}. Moreover, due to \Cref{lem:maximin}, in every preference in $\PP_2$, for every alternative $a\in\AA\setminus\DD$, there are at least $10$ alternatives from \DD in the immediate $10$ positions on both left and right of $a$, and the distance of every $d\in\DD$ from every $a\in\AA\setminus\DD$ is less than $10$ at most once in $\PP_2$. This finishes the description of \PP. We now claim that the two instances are equivalent.
  
%
%
%
%
 
 \begin{table*}
  \centering
  \begin{adjustbox}{max width=\textwidth}
  \begin{tabular}{|c|c|c|c|}\hline
   Edges in weighted majority graph & weight in \PP & weight in \QQ & Remark \\\hline\hline
   
   $(b,c)$ & $10$ & $10$ & Maximin score of $c$ is $-10$ in \PP and \QQ\\\hline
   
   \makecell{$(w_i,a_i), (w_i,\bar{a}_i) \forall i\in[n]$} & $8$ & $12$ & \makecell{Maximin scores of $a_i, \bar{a}_i, i\in[n]$ are $-8$ in \PP\\ Maximin scores of $a_i, \bar{a}_i, i\in[n]$ are $-12$ in \QQ}\\\hline
   
   \makecell{$(a_i,z_i), (\bar{a}_i,z_i) \forall i\in[n]$} & $8$ & $-$ & \makecell{Maximin scores of $z_i, i\in[n]$ are $-8$ in \PP}\\\hline
   
   \makecell{If $g(x_i)=1$, then $(\bar{a}_i,z_i)$; else $(a_i,z_i), i\in[n]$}& $-$ & $12$ &\makecell{Maximin score of $z_i$ is $-12$ in \QQ}\\\hline
   
   \makecell{$\forall j\in[m]$ if $C_j=(l_1\vee l_2\vee l_3)$, then\\ $(f(l_k),y_j)\forall k\in[3]$}& $10$ & $-$ & \makecell{Maximin score of $y_j$ is $-10$ in \PP}\\\hline
   
   \makecell{$\forall j\in[m]$ if $C_j=(l_1\vee l_2\vee l_3)$, then\\ $(f(l_k),y_j)$ for some $ k\in[3]$}& $-$ & $12$ & \makecell{Maximin score of $y_j$ is $-12$ in \QQ}\\\hline
   
   $(y_1,b), (b,w_i), \forall i\in[n]$& $12$ & $12$ & \makecell{Maximin scores of $b, w_i, i\in[n]$\\ are $-12$ in \PP and \QQ}\\\hline
   
   $(a,d)\forall a\in\AA\setminus\DD, d\in\DD$ & $10mn$ & $10mn$ & Maximin score of $d$ is $-10mn$ in \PP and \QQ\\\hline
   Any edge not mentioned above & $0$& $-$ & Weights are same in both \PP and \QQ \\\hline
  \end{tabular}
  \end{adjustbox}
  \caption{Weighted majority graph for \PP in \Cref{thm:ob_maximin_swap}.}\label{tbl:maximin_score_all}
 \end{table*}

 In one direction, let us assume that the \TSAT instance is a \YES instance with a satisfying assignment $g:\XX\longrightarrow\{0,1\}$. Let us consider the following profile \QQ where every preference is obtained from the corresponding preference in \PP by performing at most $1$ swap. The preferences in \QQ which corresponds to $\PP_1$ are as follows.
 
 \begin{enumerate}[(I)]
  \item For every $i\in[n]$, we have the following preferences.
  \begin{itemize}
   \item If $g(x_i)=1$, then
  \begin{itemize}
   \item $2$ copies of $z_i\suc w_i\suc a_i\suc \text{others}$
   \item $2$ copies of $\bar{a}_i\suc z_i\suc w_i\suc \text{others}$
  \end{itemize}
  \item Else
  \begin{itemize}
   \item $2$ copies of $a_i\suc z_i\suc w_i\suc \text{others}$
   \item $2$ copies of $z_i\suc w_i\suc \bar{a}_i\suc \text{others}$
  \end{itemize}
  \end{itemize}
  \item For every $j\in[m]$, if $C_j=(l_1\vee l_2\vee l_3)$ and $g$ makes $l_1=1$ (we can assume by renaming), then we have
  \begin{itemize}
   \item $f(l_1)\suc y_j\suc h(l_1)\suc\text{others}$
   \item $y_j\suc h(l_2)\suc f(l_2)\suc\text{others}$
   \item $y_j\suc h(l_3)\suc f(l_3)\suc\text{others}$
  \end{itemize}
 \end{enumerate}
  
 The preferences in $\PP_2$ remain unchanged in \PP and \QQ. It can be verified that $c$ is the unique maximin winner in \QQ. We describe the weighted majority graph for \QQ in \Cref{tbl:maximin_score_all} which shows that $c$ is the unique maximin winner in \QQ.
 
%
 
 In the other direction, let $\QQ=\QQ_1\cup\QQ_2$ be a profile where $c$ is the unique maximin winner, the sub-profile $\QQ_k$ corresponds to $\PP_k$ for $k\in[2]$, and every preference in \QQ is at most $1$ swap away from its corresponding preference in \PP. We begin with the observation that, since in every preference in \QQ, for every alternative $a\in\AA\setminus\DD$, there are at least $10$ alternatives from \DD in the immediate $10$ positions on both left and right of $a$ and, for every alternative $d\in\DD$, there is at most two preferences in $\PP_2$ where the distance of $d$ from any alternative in $\AA\setminus\DD$ is at most $10$, performing any one swap in any preference in $\PP_2$ leaves the maximin score of every alternative in $\AA\setminus\DD$ unchanged. Also, it is clear that any alternative in \DD can never win by performing at most one swap per preference in \PP. So, let us assume without loss of generality, that $\QQ_2=\PP_2$. Hence, the weighted majority graph of $\PP_1\cup\QQ_2$ is also as described in the second column of \Cref{tbl:maximin_score_all}. Since, there are at least $10$ alternatives from \DD in the immediate $10$ positions on both left and right of $c$ in every preference in $\PP_1$, the maximin score of $c$ is $-10$ in \QQ. Hence, for $c$ to win uniquely, $z_i$ must be preferred over either $a_i$ or $\bar{a}_i$ in at least twice of the $4$ preferences in $\PP_1$ where $z_i$ appears at the first position. Let us now consider the following assignment $g:\{x_i: i\in[n]\}\longrightarrow\{0,1\}$ defined as: $g(x_i)=0$ if $\bar{a}_i$ is preferred over $z_i$ twice among the preferences in $\QQ_1$ which correspond to the $4$ preferences in $\PP_1$ where $z_i$ appears at the first position; otherwise $g(x_i)=1$. We claim that $g$ is a satisfying assignment for the \TSAT instance. Suppose not, then let us assume that $g$ does not satisfy $C_j=(l_1\vee l_2\vee l_3)$ for some $j\in[m]$. We first observe that, since $C_j$ is not satisfied, from the definition of $g$, for every $k\in[3]$, the maximin score of $f(l_k)$ is $-8$ in the profile $\QQ_1^k\cup\QQ_2$, where $\QQ_1^k$ contains every preference in $\QQ_1$ except the $2$ preferences corresponding to the preferences $\suc_1, \suc_2$ in $\PP_1$ where $f(l_k)$ appears immediately after $y_t$ for some $t\in[m]$ and contains $\suc_1$ and $\suc_2$. Now, for $c$ to win uniquely, there must exist a $k\in[3]$ such that $f(l_k)$ is preferred over $y_j$ in at least one of the three preferences in $\QQ_1$ which corresponds to the $3$ preferences corresponding to $j$. Let us assume without loss of generality that $f(l_1)$ is preferred over $y_j$ in one of the three preferences in $\QQ_1$ which corresponds to the $3$ preferences corresponding to $j$. Then the maximin score of $f(l_1)$ in $\QQ$ is at least $-10$ which contradicts our assumption that $c$ is the unique maximin winner in \QQ. Hence, $g$ is a satisfying assignment for the \TSAT instance\longversion{ and thus the \TSAT instance is a \YES instance}.
 
 The exact same reduction and analogous proof proves the result for the footrule distance and the maximum displacement distance.
\end{proof}

\shortversion{We also prove the following results by reducing from the \TSAT problem.}
\longversion{We now prove the result for the Copeland$^\alpha$ voting rule for any $\alpha\in[0,1]$.}

\begin{theorem}\label{thm:ob_copeland_swap}
 Let $\alpha\in[0,1]$. Then the \OB problem is \NPC for the Copeland$^\alpha$ voting rule for the swap distance and the maximum displacement distance even when $\delta=1$. Hence, the \OB problem is \NPC for the Copeland$^\alpha$ voting rule for the footrule distance even when $\delta=2$.
\end{theorem}

\begin{proof}
 Let us first prove the result for the swap distance. The \OB problem for the Copeland$^\alpha$ voting rule for the swap distance is clearly in \NP. To prove \NP-hardness, we reduce from \TSAT to \OB. Let $(\XX=\{x_i: i\in[n]\},\CC=\{C_j: j\in[m]\})$ be an arbitrary instance of \TSAT. Before presenting the \OB instance that we consider we make an important remark. We will see that the number of preferences in our \OB instance is odd (which happens if and only if margins of all pairwise elections are odd integers) and thus our proof goes through for any value of $\alpha$. We now consider the following instance $(\AA,\PP,c,\delta=1)$ of \OB.
 \begin{align*}
  \AA &= \{a_i, \bar{a}_i, z_i: i\in[n]\} \\
  &\cup \{c, w\} \cup \{y_j: j\in[m]\}\cup \DD, \text{ where } |\DD|=10m^8 n^8
 \end{align*}
 
 We construct the profile \PP which is a disjoint union of two profiles, namely, $\PP_1$ and $\PP_2$. We first describe $\PP_1$ below. While describing the preferences in $\PP_1$, whenever we say `others' or `for some alternative in \DD' or `for some subset of \DD', the unspecified alternatives are assumed to be arranged in such a way that, for every unspecified alternative $a\in\AA\setminus\DD$, there are at least $10$ alternatives from \DD in the immediate $10$ positions on both left and right of $a$. We also ensure that any alternative in \DD appears within top $10mn$ positions at most once in $\PP_1$ whereas every alternative in \AA appears within top $10mn$ position in every preference in $\PP_1$. This is possible $|\DD|=10m^8 n^8, |\AA\setminus\DD|=3n+m+2,$ and $|\PP_1|=2n+3m$. Let $f$ be a function defined on the set of literals as $f(x_i)=a_i$ and $f(\bar{x}_i)=\bar{a}_i$ for every $i\in[n]$.

 \begin{enumerate}[(I)]
  \item For every $i\in[n]$, we have the following preferences.
  \begin{itemize}
   \item $z_i\suc a_i\suc w\suc \text{others}$
   \item $z_i\suc \bar{a}_i\suc w\suc \text{others}$
  \end{itemize}
  \item For every $j\in[m]$ if $C_j=(l_1\vee l_2\vee l_3)$, we have the following preferences.
  \begin{itemize}
   \item $y_j\suc f(l_1)\suc \text{others}$
   \item $y_j\suc f(l_2)\suc \text{others}$
   \item $y_j\suc f(l_3)\suc \text{others}$
  \end{itemize}
 \end{enumerate}
 
 Due to \Cref{lem:maximin}, there exists a profile $\PP_2$ consisting of $\text{poly}(m,n)$ preferences such that the weighted majority graph \GG of the profile $\PP=\PP_1\cup\PP_2$ has the following properties. There is no tie in \GG. There exists a positive integer $\NN<|\AA|$ such that,
 
 \begin{enumerate}[(i)]
  \item For every $i\in[n]$, $a_i$ and $\bar{a}_i$ defeat exactly $\NN-2$ alternatives including $w$ and lose against the remaining alternatives.
  
  \item For every $C_j=(l_1\vee l_2\vee l_3), j\in[m]$, then $y_j$ defeats exactly \NN alternatives including the alternatives $f(l_1), f(l_2), f(l_3)$ and loses against the remaining alternatives.
  
  \item For every $i\in[n]$, $z_i$ defeats exactly \NN alternatives including the alternatives $a_i$ and $\bar{a}_i$ for every $i\in[n]$ and loses against the remaining alternatives.
  
  \item The alternative $c$ defeats exactly \NN alternatives and loses against the remaining alternatives.
  
  \item The alternative $w$ and every alternative in \DD defeat at most $\NN-1$ alternatives and lose against the remaining alternatives.
  
  \item The weight of all the edges described above is $10m^3 n^3 + 1$ except the following. The weight of the edges $(z_i, a_i), (z_i, \bar{a}_i), (a_i,w), (\bar{a}_i, w)$ are $1$ for every $i\in[n]$. The weight of the edges $(y_j, f(l_1)), (y_j, f(l_2)), (y_j, f(l_3))$ are $1$ for every $j\in[m]$ where $C_j=(l_1\vee l_2\vee l_3)$. We observe that all pairwise margins are odd integers and thus we have an odd number of preferences in the constructed preference profile. In particular, all the arguments made below work equally well for any value of $\alpha$.
 \end{enumerate}

 Moreover, due to \Cref{lem:maximin} (using $K=30$ in \Cref{lem:maximin}), in every preference in $\PP_2$, for every alternative $a\in\AA\setminus\DD$, there are at least $10$ alternatives from \DD in the immediate $10$ positions on both left and right of $a$, and the distance of every $d\in\DD$ from every $a\in\AA\setminus\DD$ is less than $10$ at most once in $\PP_2$. This finishes the description of \PP. We now claim that the two instances are equivalent.

 In one direction, let us assume that the \TSAT instance is a \YES instance with satisfying assignment $g:\XX\longrightarrow\{0,1\}$. Let us consider the following profile \QQ where every preference is obtained from the corresponding preference in \PP by performing at most $1$ swap. The preferences in \QQ which corresponds to $\PP_1$ are as follows.
 
 \begin{enumerate}[(I)]
  \item For every $i\in[n]$, we have the following preferences.
  \begin{itemize}
   \item If $g(x_i)=1$, then
  \begin{itemize}
   \item $z_i\suc w\suc a_i\suc \text{others}$
   \item $\bar{a}_i\suc z_i\suc w\suc \text{others}$
  \end{itemize}
  \item Else
  \begin{itemize}
   \item $a_i\suc z_i\suc w\suc \text{others}$
   \item $z_i\suc w\suc \bar{a}_i\suc \text{others}$
  \end{itemize}
  \end{itemize}
  \item For every $j\in[m]$, if $C_j=(l_1\vee l_2\vee l_3)$ and $g$ makes $l_1=1$ (we can assume by renaming), then we have
  \begin{itemize}
   \item $f(l_1)\suc y_j\suc \text{others}$
   \item $y_j\suc h(l_2)\suc \text{others}$
   \item $y_j\suc h(l_3)\suc \text{others}$
  \end{itemize}
 \end{enumerate}
 
 The preferences in $\PP_2$ remain unchanged in \PP and \QQ. We observe that the Copeland$^\alpha$ score of $c$ is \NN whereas the Copeland$^\alpha$ score of every alternative is at most $\NN-1$. Thus $c$ is the unique Copeland$^\alpha$ winner in \QQ.

 In the other direction, let $\QQ=\QQ_1\cup\QQ_2$ be a profile where $c$ is the unique Copeland$^\alpha$ winner, the sub-profile $\QQ_k$ corresponds to $\PP_k$ for $k\in[2]$, and every preference in \QQ is at most $1$ swap away from its corresponding preference in \PP. We begin with the observation that, due to the structure of $\PP_1$ and $\PP_2$ (from \Cref{lem:maximin}), any change in any preference in $\PP_2$ up to $1$ swap, irrespective of $\QQ_1$, leaves the Copeland$^\alpha$ score of every alternative in $\AA\setminus\DD$ unchanged. So we can assume without loss of generality that $\QQ_2=\PP_2$. Since, for every alternative $a\in\AA\setminus\DD$, there are at least $10$ alternatives from \DD in the immediate $10$ positions on both left and right of $a$ in every preference in $\PP_1$, the weight of any edge incident on $s$ in the weighted majority graph is $10m^3 n^3 + 1$, and any alternative in $\AA\setminus\{c\}$ follows immediately $c$ in at most one preference in \PP, the Copeland$^\alpha$ score of $c$ is $\NN$ in \QQ. Hence, for $c$ to win uniquely, $z_i$ must be preferred over either $a_i$ or $\bar{a}_i$ in at least $1$ of the $2$ preferences in $\PP_1$ where $z_i$ appears at the first position. Let us now consider the following assignment $g:\{x_i: i\in[n]\}\longrightarrow\{0,1\}$ defined as: $g(x_i)=1$ if $\bar{a}_i$ is preferred over $z_i$ once among the preferences in $\QQ_1$ which corresponds to the $2$ preferences in $\PP_1$ where $z_i$ appears at the first position; otherwise $g(x_i)=0$. We claim that $g$ is a satisfying assignment for the \TSAT instance. Suppose not, then let us assume that $g$ does not satisfy $C_j=(l_1\vee l_2\vee l_3)$ for some $j\in[m]$. Let $h$ be a function defined on the set of literals as $h(x_i)=z_i$ and $h(\bar{x}_i)=z_i$ for every $i\in[n]$. We first observe that, since $C_j$ is not satisfied, from the definition of $g$, for every $k\in[3]$, the alternative $f(l_k)$ defeats $h(l_k)$ in \QQ. Hence, for $c$ to become the unique Copeland$^\alpha$ winner in \QQ, $f(l_k)$ must lose to $y_j$ for every $k\in[3]$. However, this makes the Copeland$^\alpha$ score of $y_j$ in \QQ $\NN$ which contradicts ous assumption that $c$ is the unique Copeland$^\alpha$ winner in \QQ. Hence, $g$ is a satisfying assignment for the \TSAT instance and thus the \TSAT instance is a \YES instance.
 
 The exact same reduction and analogous proof proves the result for the footrule distance and the maximum displacement distance.
\end{proof}

\begin{theorem}\label{thm:ob_bucklin_swap}
 The \OB problem is \NPC for the simplified Bucklin voting rule for the swap distance even when $\delta=2$. Hence the \OB problem is \NPC for the simplified Bucklin voting rule for the footrule distance even when $\delta=4$.
\end{theorem}

\begin{proof}
 The \OB problem for the simplified Bucklin voting rule for the swap distance is clearly in \NP. To prove \NP-hardness, we reduce from \TSAT to \OB. Let $(\XX=\{x_i: i\in[n]\},\CC=\{C_j: j\in[m]\})$ be an arbitrary instance of \TSAT. Let us consider the following instance $(\AA,\PP,c,\delta=2)$ of \OB.
 \begin{align*}
  \AA &= \{a(x_i,0), a(x_i,1), a(\bar{x}_i,0), a(\bar{x}_i,1): i\in[n]\} \cup \{c\}\\
  &\cup \{w_i: i\in[n]\} \cup \{y_j: j\in[m]\} \cup \DD, \text{ where } |\DD|=10m^8 n^8
 \end{align*}
 
 We construct the profile \PP using the following function $f$. The function $f$ takes a literal and a clause as input, and outputs a value in $\{0,1,-\}$. For each literal $l$, let $C_i$ and $C_j$ with $1\le i<j\le m$ be the two clauses where $l$ appears. We define $f(l,C_i)=0, f(l,C_j)=1,$ and $f(l,C_k)=-$ for every $k\in[m]\setminus\{i,j\}$. This finishes the description of the function $f$. We are now ready to describe \PP. While describing the preferences below, whenever we say `others' or `for some alternative in \DD' or `for some subset of \DD', the unspecified alternatives are assumed to be arranged in such a way that, for every unspecified alternative $a\in\AA\setminus\DD$, there are at least $10$ alternatives from \DD in the immediate $10$ positions on both left and right of $a$. We also ensure that any alternative in \DD appears within top $10m^2n^2$ positions at most once in $\PP$ whereas every alternative in \AA appears within top $10m^2n^2$ position in every preference in $\PP$. This is possible because $|\DD|=10m^8 n^8, |\AA\setminus\DD|=5n+m+2,$ and $|\PP|=4n+6m+1$. Let us define $\WW=\{w_i: i\in[n]\}$.
 \begin{enumerate}[(I)]
  \item For every $i\in[n]$, we have\label{thm:ob_bucklin_1}
  \begin{itemize}
   \item $c\suc\DD_{mn-3}\suc w_i\suc a(x_i,0)\suc a(x_i,1)\suc d\suc \text{others}$, for some $d\in\DD$ and $\DD_{mn-3}\subset\DD$ with $|\DD_{mn-3}|=mn-3$
   \item $c\suc\DD_{mn-3}\suc w_i\suc a(\bar{x}_i,0)\suc a(\bar{x}_i,1)\suc d\suc \text{others}$, for some $d\in\DD$ and $\DD_{mn-3}\subset\DD$ with $|\DD_{mn-3}|=mn-3$
  \end{itemize}
  \item For every $j\in[m]$, if $C_j=(l_1\vee l_2\vee l_3)$, then we have\label{thm:ob_bucklin_2}
  \begin{itemize}
   \item $c\suc\DD_{mn-3}\suc y_j\suc d\suc a(l_1,f(l_1,C_j))\suc d^\pr\suc \text{others}$, for some $d, d^\pr\in\DD$ and $\DD_{mn-3}\subset\DD$ with $|\DD_{mn-3}|=mn-3$
   \item $c\suc\DD_{mn-3}\suc y_j\suc d\suc a(l_2,f(l_2,C_j))\suc d^\pr\suc \text{others}$, for some $d, d^\pr\in\DD$ and $\DD_{mn-3}\subset\DD$ with $|\DD_{mn-3}|=mn-3$
   \item $c\suc\DD_{mn-3}\suc y_j\suc d\suc a(l_3,f(l_3,C_j))\suc d^\pr\suc \text{others}$, for some $d, d^\pr\in\DD$ and $\DD_{mn-3}\subset\DD$ with $|\DD_{mn-3}|=mn-3$
  \end{itemize}
  
  \item $\DD_{2mn}\suc \text{others}$, for some $\DD_{2mn}\subset\DD$ with $|\DD_{2mn}|=2mn$ \label{thm:ob_bucklin_3}
  
  \item $2n+3m-2$ copies: $(\AA\setminus(\DD\cup\{c\}))\suc \DD \suc c$\label{thm:ob_bucklin_4}
  
  \item $\DD_{mn+1}\suc c\suc \text{ others}$, for some $\DD_{mn+1}\subset\DD$ with $|\DD_{mn+1}|=mn+1$ \label{thm:ob_bucklin_5}
  
  \item $(\AA\setminus(\DD\cup\{c\}\cup\{y_j: j\in[m]\}))\suc \DD\suc y_1\suc\cdots y_m \suc c$\label{thm:ob_bucklin_6}
 \end{enumerate}

 We now claim that the two instances are equivalent. In one direction, let us assume that the \TSAT instance is a \YES instance with satisfying assignment $g:\XX\longrightarrow\{0,1\}$. Let us consider the following profile \QQ where every preference is obtained from the corresponding preference in \PP by performing at most $2$ swaps.
 \begin{enumerate}[(I)]
  \item For every $i\in[n]$, we have
  \begin{itemize}
   \item If $g(x_i)=1$, then
   \begin{itemize}
    \item $c\suc\DD_{mn-3}\suc w_i\suc d\suc a(x_i,0)\suc a(x_i,1)\suc \text{others}$ for some $d\in\DD$ and $\DD_{mn-3}\subset\DD$ with $|\DD_{mn-3}|=mn-3$
   \item $c\suc\DD_{mn-3}\suc a(\bar{x}_i,0)\suc a(\bar{x}_i,1)\suc w_i\suc d\suc \text{others}$ for some $d\in\DD$ and $\DD_{mn-3}\subset\DD$ with $|\DD_{mn-3}|=mn-3$
   \end{itemize}
   \item Else
   \begin{itemize}
   \item $c\suc\DD_{mn-3}\suc a(x_i,0)\suc a(x_i,1)\suc w_i\suc d\suc \text{others}$ for some $d\in\DD$ and $\DD_{mn-3}\subset\DD$ with $|\DD_{mn-3}|=mn-3$
   \item $c\suc\DD_{mn-3}\suc w_i\suc d\suc a(\bar{x}_i,0)\suc a(\bar{x}_i,1)\suc \text{others}$ for some $d\in\DD$ and $\DD_{mn-3}\subset\DD$ with $|\DD_{mn-3}|=mn-3$
   \end{itemize}
  \end{itemize}
  \item For every $j\in[m]$, if $C_j=(l_1\vee l_2\vee l_3)$ and $g$ makes $l_1=1$ (we can assume by renaming), then we have
  \begin{itemize}
   \item $c\suc\DD_{mn-3}\suc d\suc a(l_1,f(l_1,C_j))\suc y_j\suc d^\pr\suc \text{others}$ for some $d, d^\pr\in\DD$ and $\DD_{mn-3}\subset\DD$ with $|\DD_{mn-3}|=mn-3$
   \item $c\suc\DD_{mn-3}\suc y_j\suc d\suc a(l_2,f(l_2,C_j))\suc d^\pr\suc \text{others}$ for some $d, d^\pr\in\DD$ and $\DD_{mn-3}\subset\DD$ with $|\DD_{mn-3}|=mn-3$
   \item $c\suc\DD_{mn-3}\suc y_j\suc d\suc a(l_3,f(l_3,C_j))\suc d^\pr\suc \text{others}$ for some $d, d^\pr\in\DD$ and $\DD_{mn-3}\subset\DD$ with $|\DD_{mn-3}|=mn-3$
  \end{itemize}
  
  \item $\DD_{2mn}\suc \text{others}$, for some $\DD_{2mn}\subset\DD$ with $|\DD_{2mn}|=2mn$
  
  \item $2n+3m-2$ copies: $(\AA\setminus\{c\})\suc \DD \suc c$
  
  \item $\DD_{mn-1}\suc c\suc \text{ others}$, for some $\DD_{mn-1}\subset\DD$ with $|\DD_{mn-1}|=mn-1$
  
  \item $(\AA\setminus(\DD\cup\{c\}\cup\{y_j: j\in[m]\}))\suc \DD\suc y_1\suc\cdots y_m \suc c$
 \end{enumerate}
 
 Only the alternative $c$ appears a majority number of times within the first $mn$ positions in \QQ. Hence $c$ is the unique simplified Bucklin winner in \QQ and thus the \OB instance is a \YES instance.
 
 In the other direction, let $\QQ$ be a profile where $c$ is the unique simplified Bucklin winner and every preference in \QQ is at most $2$ swaps away from its corresponding preference in \PP. We observe that in the preference in \QQ which corresponds to \Cref{thm:ob_bucklin_5}, the alternative $c$ moves to its left by $2$ positions since otherwise, irrespective of \QQ (subject to the condition that the swap distance of every preference in \QQ is at most $2$ from its corresponding preference in \PP), $c$ appears within the first $mn$ positions only $(2n+3m)$ times in \QQ which is not a majority number of times whereas every alternative in \WW appears a majority number of times within the first $(mn+1)$ positions contradicting our assumption that $c$ is the unique simplified Bucklin winner in \QQ. Hence, $c$ appears $(2n+3m+1)$ (which is a majority) number of times within the first $mn$ positions in \QQ and does not appear a majority number of times within the first $mn-1$ positions in \QQ. Now, since $c$ is the unique simplified Bucklin winner in \QQ, no alternative other than $c$ appears a majority number of times within the first $mn$ positions in \QQ. Since, every alternative $w_i\in \WW$ appears $(2n+3m+1)$ number of times within the first $mn-1$ positions in \PP, $w_i$ must be moved in \QQ to its right by $2$ positions in at least one of the two preferences in \PP where $w_i$ appears at position $(mn-1)$.  We now consider the following assignment $g:\XX\longrightarrow\{0,1\}$ -- for every $i\in[n]$, $g(x_i)=0$ if the preference in \QQ corresponding to the preference $c\suc \DD_{mn-3}\suc w_i\suc a(x_i,0)\suc a(x_i,1)\suc d\suc \text{others}$ in \PP moves $w_i$ to its right by two positions; otherwise $g(x_i)=1$. We claim that $g$ is a satisfying assignment for the \TSAT instance. Suppose not, then let us assume that $g$ does not satisfy $C_j=(l_1\vee l_2\vee l_3)$ for some $j\in[m]$. We first observe that $y_j$ also appears $(2n+3m+1)$ number of times within the first $mn-1$ positions in \PP. Hence, $y_j$ must be moved in \QQ to its right by $2$ positions in at least one of the three preferences in \PP where $y_j$ appears at position $(mn-1)$. However, this implies that at least one alternative in $\{a(l_i, f(l_i, C_j)):i\in[3]\}$ appears at least $(2n+3m+1)$ times within the first $mn$ positions in \QQ due to the definition of $g$. This contradicts our assumption that $c$ is the unique simplified Bucklin winner in \QQ. Hence $g$ is a satisfying assignment and thus the \TSAT instance is a \YES instance.
 
 The proof for the footrule distance follows from the observation that, for any two preferences $\suc_1, \suc_2\in\LL(\AA)$, if we have $d_{\text{footrule}}(\suc_1,\suc_2)=4$, then we have $d_{\text{swap}}(\suc_1,\suc_2)=2$.
\end{proof}

\begin{theorem}\label{thm:ob_bucklin_maxdis_hard}\shortversion{[\star]}
 The \ODB problem is \NPC for the simplified Bucklin voting rule for the maximum displacement distance even when $\delta_i=2$ for every preference $i$.
\end{theorem}

\longversion{
\begin{proof}
 The proof is analogous to the proof of \Cref{thm:kapp_max_hard}.
\end{proof}
}

\begin{theorem}\label{thm:ob_bucklin}\shortversion{[\star]}
 \longversion{The \OB problem is \NPC for the Bucklin voting rule for the swap distance and maximum displacement distance even when $\delta=1$. Hence, The \OB problem is \NPC for the Bucklin voting rule for the footrule distance even when $\delta=2$.}
 \shortversion{The \OB problem is \NPC for the Bucklin voting rule for the swap distance and maximum displacement distance even when $\delta=1$ and thus for the footrule distance even when $\delta=2$.}
\end{theorem}
 
\longversion{
\begin{proof}
 Let us first prove the result for the maximum displacement distance. The \OB problem for the Bucklin voting rule for the maximum displacement distance is clearly in \NP. To prove \NP-hardness, we reduce from \TSAT to \OB. Let $(\XX=\{x_i: i\in[n]\},\CC=\{C_j: j\in[m]\})$ be an arbitrary instance of \TSAT. Let us consider the following instance $(\AA,\PP,c,\delta=1)$ of \OB.
 \begin{align*}
  \AA &= \{a_i, \bar{a}_i, z_i: i\in[n]\} \cup \{c\}\\
  &\cup \{y_j, e_j: j\in[m]\} \cup \DD, \text{ where } |\DD|=10m^8 n^8
 \end{align*}
 
 We construct the profile \PP using the following function $f$. The function $f$ takes a literal and a clause as input, and outputs a value in $\{0,1,-\}$. For each literal $l$, let $C_i$ and $C_j$ with $1\le i<j\le m$ be the two clauses where $l$ appears. We define $f(l,C_i)=0, f(l,C_j)=1,$ and $f(l,C_k)=-$ for every $k\in[m]\setminus\{i,j\}$. This finishes the description of the function $f$. We are now ready to describe \PP. While describing the preferences below, whenever we say `others' or `for some alternative in \DD' or `for some subset of \DD', the unspecified alternatives are assumed to be arranged in such a way that, for every unspecified alternative $a\in\AA\setminus\DD$, there are at least $10$ alternatives from \DD in the immediate $10$ positions on both left and right of $a$. We also ensure that any alternative in \DD appears within top $10mn$ positions at most once in $\PP$ whereas every alternative in \AA appears within top $10m^2n^2$ position in every preference in $\PP$. This is possible because $|\DD|=10m^8 n^8, |\AA\setminus\DD|=3n+2m+2,$ and $|\PP|\le 10m^3 n^3$. Let $h$ be a function defined on the set of literals as $h(x_i)=a_i$ and $h(\bar{x}_i)=\bar{a}_i$ for every $i\in[n]$. Let $k=10(m+n)$ and $N^\pr=2n+5m$. For any integer $s$ with $1\le s\le 4$, let $Y_s = \{y_j: j\in[m]$ and if $C_j=(l_1\vee l_2\vee l_3)$ and $f(l_r, C_j)=0$ for exactly $s-1$ many $r\in\{1,2,3\}\}$. The profile $\PP$ is the disjoint union of two profiles $\PP_1$ and $\PP_2$. We first describe $\PP_1$ below.
 
 \begin{enumerate}[(I)]
  \item For every $i\in[n]$, we have\label{thm:ob_bucklin_tight_1}
  \begin{itemize}
   \item $c\suc\DD_{k-3}\suc z_i\suc a_i\suc \text{others}$, for some $\DD_{k-3}\subset\DD$ with $|\DD_{k-3}|=k-3$
   \item $c\suc\DD_{k-3}\suc z_i\suc \bar{a}_i\suc \text{others}$, for some $\DD_{k-3}\subset\DD$ with $|\DD_{k-3}|=k-3$
  \end{itemize}
  \item For every $j\in[m]$, if $C_j=(l_1\vee l_2\vee l_3)$, then we have\label{thm:ob_bucklin_tight_2}
  \begin{itemize}
   \item $c\suc\DD_{k-3}\suc y_j\suc e_j\suc \text{others}$, for some $\DD_{k-3}\subset\DD$ with $|\DD_{k-3}|=k-3$
   \item $c\suc\DD_{k-2}\suc y_j\suc e_j\suc \text{others}$, for some $\DD_{k-2}\subset\DD$ with $|\DD_{k-2}|=k-2$
   \item For every $r\in[3]$\\
   If $f(l_r, C_j)=0$, then\\
   $c\suc\DD_{k-3}\suc y_j\suc h(l_r)\suc \text{others}$, for some $\DD_{k-3}\subset\DD$ with $|\DD_{k-3}|=k-3$\\
   Otherwise\\
   $c\suc\DD_{k-2}\suc y_j\suc h(l_r)\suc \text{others}$, for some $\DD_{k-2}\subset\DD$ with $|\DD_{k-2}|=k-2$
  \end{itemize}
  
  \item $5$ copies: $\DD_k \suc c\suc \text{others}$, for some $\DD_k\subset\DD$ with $|\DD_k|=k$\label{thm:ob_bucklin_tight_3}
  
  \item For every $i\in[n]$\label{thm:ob_bucklin_tight_4}
  \begin{itemize}
   \item $N^\pr-1$ copies: $\DD_{k-3} \suc z_i\suc \text{others}$, for some $\DD_{k-3}\subset\DD$ with $|\DD_{k-3}|=k-3$
   \item $N^\pr-1$ copies: $\DD_{k-3} \suc a_i\suc \text{others}$, for some $\DD_{k-3}\subset\DD$ with $|\DD_{k-3}|=k-3$
   \item $N^\pr-1$ copies: $\DD_{k-3} \suc \bar{a}_i\suc \text{others}$, for some $\DD_{k-3}\subset\DD$ with $|\DD_{k-3}|=k-3$
   \item $3$ copies: $\DD_{k-2} \suc a_i\suc \text{others}$, for some $\DD_{k-2}\subset\DD$ with $|\DD_{k-2}|=k-2$
   \item $3$ copies: $\DD_{k-2} \suc \bar{a}_i\suc \text{others}$, for some $\DD_{k-2}\subset\DD$ with $|\DD_{k-2}|=k-2$
  \end{itemize}
  
  \item For every $j\in[m]$\label{thm:ob_bucklin_tight_5}
  \begin{itemize}
   \item $N^\pr+3$ copies: $\DD_{k-2} \suc e_j\suc \text{others}$, for some $\DD_{k-2}\subset\DD$ with $|\DD_{k-2}|=k-2$
  \end{itemize}

  \item For every $s\in[4]$, for every $y_j\in Y_s$\label{thm:ob_bucklin_tight_6}
  \begin{itemize}
   \item $N^\pr-s+1$ copies: $\DD_{k-3} \suc y_j\suc \text{others}$, for some $\DD_{k-3}\subset\DD$ with $|\DD_{k-3}|=k-3$
   \item $s-1$ copies: $\DD_{k-2} \suc y_j\suc \text{others}$, for some $\DD_{k-2}\subset\DD$ with $|\DD_{k-2}|=k-2$
  \end{itemize}
 \end{enumerate}
 
 This finishes the description of $\PP_1$. Let the number of preferences in $\PP_1$ be $N^\prr$. Let $\PP_2$ be a profile consisting of $N^\prr-2N^\pr$ copies of $(\AA\setminus\DD)\suc\DD$. This finishes the description of $\PP_2$ and thus of \PP. Let $N$ denotes the total number of preferences in \PP. Then we have $N=2(N^\prr-N^\pr)$. We observe that $N$ is an even integer. This finishes the description of the reduced \OB instance. \Cref{fig:buck1} shows the number of times every alternative appears within the first $k-2, k-1, k,$ and $k+1$ times in \PP.
 
 \begin{table*}[!htbp]
  \centering
  \begin{tabular}{c|cccc}\hline
   \multirow{2}{*}{Alternatives} & \multicolumn{4}{c}{Number of times it appears within first}\\\cline{2-5}
   & $k-2$ positions& $k-1$ positions& $k$ positions& $k+1$ positions\\\hline\hline
   
   $c$ & $(\nfrac{N}{2})$ & $(\nfrac{N}{2})$ & $(\nfrac{N}{2})$ & $(\nfrac{N}{2})+5$ \\
   
   $z_i, i\in[n]$ & $(\nfrac{N}{2})-1$ & $(\nfrac{N}{2})+1$ & $(\nfrac{N}{2})+1$ & $(\nfrac{N}{2})+1$ \\
   
   $a_i, \bar{a}_i, i\in[n]$ & $(\nfrac{N}{2})-1$ & $(\nfrac{N}{2})+2$ & $(\nfrac{N}{2})+4$ &  $(\nfrac{N}{2})+5$ \\
   
   $y_j, y_j\in Y_s, s\in[4]$ & $(\nfrac{N}{2})-s+1$ & $(\nfrac{N}{2})+s$ & $(\nfrac{N}{2})+5$ & $(\nfrac{N}{2})+5$ \\
   
   $e_j, j\in[m]$ & $0$ & $(\nfrac{N}{2})+3$ & $(\nfrac{N}{2})+4$ & $(\nfrac{N}{2})+5$  \\
   
   $d\in\DD$ & $\le1$ & $\le1$ & $\le1$ & $\le1$ \\\hline
  \end{tabular}
  \caption{Number of times every alternative appears within first $k-2, k-1, k,$ and $k+1$ times in \PP.}\label{fig:buck1}
 \end{table*}

 We now claim that the two instances are equivalent. In one direction, let us assume that the \TSAT instance is a \YES instance with satisfying assignment $g:\XX\longrightarrow\{0,1\}$. Let us consider the following profile \QQ where the maximum displacement distance of every preference in \QQ from its corresponding preference in \PP is at most $1$.

 \begin{enumerate}[(I)]
  \item For every $i\in[n]$, we have
  \begin{itemize}
   \item If $g(x_i)=0$, then
   \begin{itemize}
    \item $c\suc\DD_{k-3}\suc a_i\suc z_i\suc \text{others}$, for some $\DD_{k-3}\subset\DD$ with $|\DD_{k-3}|=k-3$
    \item $c\suc\DD_{k-3}\suc z_i\suc d\suc \bar{a}_i\suc \text{others}$, for some $d\in\DD, \DD_{k-3}\subset\DD$ with $|\DD_{k-3}|=k-3$    
   \end{itemize}
   \item If $g(x_i)=1$, then
   \begin{itemize}
    \item $c\suc\DD_{k-3}\suc z_i\suc d\suc a_i\suc \text{others}$, for some $d\in\DD, \DD_{k-3}\subset\DD$ with $|\DD_{k-3}|=k-3$
    \item $c\suc\DD_{k-3}\suc \bar{a}_i\suc z_i\suc \text{others}$, for some $\DD_{k-3}\subset\DD$ with $|\DD_{k-3}|=k-3$    
   \end{itemize}
  \end{itemize}
  \item For every $j\in[m]$, if $C_j=(l_1\vee l_2\vee l_3)$, then we have
  \begin{itemize}
   \item If there exists $t\in[3]$ with $f(l_t, C_j)=0$ and $g(l_t)=1$, then
    \begin{itemize}
    \item $c\suc\DD_{k-3}\suc y_j\suc d\suc e_j\suc \text{others}$, for some $d\in\DD, \DD_{k-3}\subset\DD$ with $|\DD_{k-3}|=k-3$
    \item $c\suc\DD_{k-2}\suc e_j\suc y_j\suc \text{others}$, for some $\DD_{k-2}\subset\DD$ with $|\DD_{k-2}|=k-2$
    \item $c\suc\DD_{k-3}\suc h(l_t)\suc y_j\suc \text{others}$, for some $\DD_{k-3}\subset\DD$ with $|\DD_{k-3}|=k-3$
    \item For every $r\in[3], r\ne t$\\
    If $f(l_r, C_j)=0$, then\\
    $c\suc\DD_{k-3}\suc y_j\suc h(l_r)\suc \text{others}$, for some $\DD_{k-3}\subset\DD$ with $|\DD_{k-3}|=k-3$\\
    Otherwise\\
    $c\suc\DD_{k-2}\suc y_j\suc h(l_r)\suc \text{others}$, for some $\DD_{k-2}\subset\DD$ with $|\DD_{k-2}|=k-2$
    \end{itemize}
    
   \item If there exists $t\in[3]$ with $f(l_t, C_j)=1$ and $g(l_t)=1$, then
    \begin{itemize}
    \item $c\suc\DD_{k-3}\suc e_j\suc y_j\suc \text{others}$, for some $\DD_{k-3}\subset\DD$ with $|\DD_{k-3}|=k-3$
    \item $c\suc\DD_{k-2}\suc y_j\suc e_j\suc \text{others}$, for some $\DD_{k-2}\subset\DD$ with $|\DD_{k-2}|=k-2$
    \item $c\suc\DD_{k-2}\suc h(l_t)\suc y_j\suc \text{others}$, for some $\DD_{k-2}\subset\DD$ with $|\DD_{k-2}|=k-2$
    \item For every $r\in[3], r\ne t$\\
    If $f(l_r, C_j)=0$, then\\
    $c\suc\DD_{k-3}\suc y_j\suc h(l_r)\suc \text{others}$, for some $\DD_{k-3}\subset\DD$ with $|\DD_{k-3}|=k-3$\\
    Otherwise\\
    $c\suc\DD_{k-2}\suc y_j\suc h(l_r)\suc \text{others}$, for some $\DD_{k-2}\subset\DD$ with $|\DD_{k-2}|=k-2$
    \end{itemize}
  
  \end{itemize}
  \item $5$ copies: $\DD_{k-1} \suc c\suc \text{others}$, for some $\DD_{k-1}\subset\DD$ with $|\DD_{k-1}|=k-1$
  
  \item For every $i\in[n]$
  \begin{itemize}
   \item $N^\pr-1$ copies: $\DD_{k-2} \suc z_i\suc \text{others}$, for some $\DD_{k-2}\subset\DD$ with $|\DD_{k-2}|=k-2$
   \item $N^\pr-1$ copies: $\DD_{k-2} \suc a_i\suc \text{others}$, for some $\DD_{k-2}\subset\DD$ with $|\DD_{k-2}|=k-2$
   \item $N^\pr-1$ copies: $\DD_{k-2} \suc \bar{a}_i\suc \text{others}$, for some $\DD_{k-2}\subset\DD$ with $|\DD_{k-2}|=k-2$
   \item $3$ copies: $\DD_{k-1} \suc a_i\suc \text{others}$, for some $\DD_{k-1}\subset\DD$ with $|\DD_{k-1}|=k-1$
   \item $3$ copies: $\DD_{k-1} \suc \bar{a}_i\suc \text{others}$, for some $\DD_{k-1}\subset\DD$ with $|\DD_{k-1}|=k-1$
  \end{itemize}
  
  \item For every $j\in[m]$
  \begin{itemize}
   \item $N^\pr+3$ copies: $\DD_{k-1} \suc e_j\suc \text{others}$, for some $\DD_{k-1}\subset\DD$ with $|\DD_{k-1}|=k-1$
  \end{itemize}

  \item For every $s\in[4]$, for every $y_j\in Y_s$
  \begin{itemize}
   \item $N^\pr-s+1$ copies: $\DD_{k-2} \suc y_j\suc \text{others}$, for some $\DD_{k-2}\subset\DD$ with $|\DD_{k-2}|=k-2$
   \item $s-1$ copies: $\DD_{k-1} \suc y_j\suc \text{others}$, for some $\DD_{k-1}\subset\DD$ with $|\DD_{k-1}|=k-1$
  \end{itemize}
 \end{enumerate}
 
 \Cref{fig:buck2} shows the number of times every alternative appears within the first $k-1$ and $k$ times in \QQ which proves that $c$ is the unique Bucklin winner in \QQ and thus the \OB instance is a \YES instance.
 
 \begin{table*}[!htbp]
  \centering
  \begin{tabular}{c|cc}\hline
   \multirow{2}{*}{Alternatives} & \multicolumn{2}{c}{Number of times it appears within first}\\\cline{2-3}
   & $k-1$ positions& $k$ positions\\\hline\hline
   
   $c$ & $(\nfrac{N}{2})$ & $(\nfrac{N}{2})+5$ \\
   
   $z_i, i\in[n]$ & $(\nfrac{N}{2})$ & $(\nfrac{N}{2})+1$ \\
   
   $a_i, \bar{a}_i, i\in[n]$ & $\le(\nfrac{N}{2})$ & $\le(\nfrac{N}{2})+4$ \\
   
   $y_j, j\in[m]$ & $(\nfrac{N}{2})$ & $(\nfrac{N}{2})+4$ \\
   
   $e_j, j\in[m]$ & $\le 1$ & $(\nfrac{N}{2})+4$ \\
   
   $d\in\DD$ & $\le1$ & $\le1$ \\\hline
  \end{tabular}
  \caption{Number of times every alternative appears within first $k-1$ and $k$ times in \QQ in the forward direction of the proof of \Cref{thm:ob_bucklin}.}\label{fig:buck2}
 \end{table*}

 In the other direction, let us assume that there exists a profile \QQ such that the maximum displacement distance of every preference in \QQ is at most $1$ from its corresponding preference in \PP and $c$ is the unique Bucklin winner in \QQ. We first observe that, irrespective of \QQ (subject to the condition that its maximum displacement distance from \PP is at most $1$), any alternative in \DD does not appear a majority number of times within the first $2k$ positions in \QQ. Nest we can assume without loss of generality that in the preference in \QQ which corresponds to \Cref{thm:ob_bucklin_tight_3}, the alternative $c$ moves to its left by $1$ position since any alternative to the left of $c$ in \Cref{thm:ob_bucklin_tight_3} belongs to \DD. Since in any preference in \QQ other than \Cref{thm:ob_bucklin_tight_3}, $c$ never appears in any position in $\{\el\in\NB: k-3\le \el\le k+3\}$, $c$ does not appear a majority number of times within the first $k-1$ positions (it actually appears $(\nfrac{N}{2})$ times) in \QQ. However, $c$ appears within the first $k$ positions $(\nfrac{N}{2})+5$ times in \QQ. Since, in every preference in \Cref{thm:ob_bucklin_tight_4,thm:ob_bucklin_tight_5,thm:ob_bucklin_tight_6}, every alternative in $\AA\setminus\DD$ has at least 10 alternatives from \DD in its immediate left and right, we can assume without loss of generality that every alternative in $\AA\setminus(\DD\cup\{c\})$ moves to its right by $1$ position in every preference in \QQ corresponding to the preferences in \Cref{thm:ob_bucklin_tight_4,thm:ob_bucklin_tight_5,thm:ob_bucklin_tight_6}. Now the alternative $z_i, i\in[n]$ appears $(\nfrac{N}{2})+1$ times within the first $k-1$ positions, $y_j, j\in[m]$ appears $(\nfrac{N}{2})+1$ and $(\nfrac{N}{2})+5$ times within the first $k-1$ and $k$ positions respectively, $a_i, \bar{a}_i, i\in[n]$ appear $(\nfrac{N}{2})-1$ and $(\nfrac{N}{2})+4$ times within the first $k-1$ and $k$ positions respectively, and $e_j, j\in[m]$ appears $1$ and $(\nfrac{N}{2})+4$ times within the first $k-1$ and $k$ positions. We now observe that there are exactly two preferences in \PP where $z_i$ appears at the $(k-1)$-th position. Hence, for $c$ to be the unique Bucklin winner in \QQ, $z_i$ must move to its right by $1$ position in at least one preference among the two preferences in \PP where it appears at $(k-1)$-th position. We now consider the following assignment $g:\XX\longrightarrow\{0,1\}$ -- for every $i\in[n]$, $g(x_i)=0$ if there exists a preference in \QQ where $z_i$ appears at the $k$-th position and $a_i$ appears at the $(k-1)$-th position; otherwise we define $g(x_i)=1$. We claim that $g$ is a satisfying assignment for the \TSAT instance. Suppose not, then let us assume that $g$ does not satisfy $C_j=(l_1\vee l_2\vee l_3)$ for some $j\in[m]$. Suppose $y_j\in Y_s$ for some $s\in\{1,2,3,4\}$. Then, among the $5$ preferences in \Cref{thm:ob_bucklin_tight_2} corresponding to $y_j$ (let us call them $\QQ_j$ and $\PP_j$ respectively in \QQ and in \PP), the alternative $y_j$ appears $s$ times at the $(k-1)$-th position and $5-s$ times at the $k$-th position. Now it follows from the definition of $g$ that every alternative in $\{h(l_1), h(l_2), h(l_3)\}$ appears $(\nfrac{N}{2})$ and $(\nfrac{N}{2})+4$ times within the first $k-1$ and $k$ positions in $(\QQ\setminus\QQ_j)\cup \PP_j$ and thus cannot move to their left in any preference in $\PP_j$. Now, to make $c$ the unique Bucklin winner of \QQ, the alternative $y_j$ must move to its right by one position each in at least one preferences where it appears at the $(k-1)$-th position and where it appears at the $k$-th position. Hence, $e_j$ must move to its left in both the preferences in $\PP_j$ where it appears at the immediate right of $y_j$. However, this makes $e_j$ appear $(\nfrac{N}{2})+5$ times within the first $k$ preferences in \QQ which contradicts our assumption that $c$ is the unique Bucklin winner in \QQ. Hence $g$ is a satisfying assignment and thus the \TSAT instance is a \YES instance.
 
 Now the result for the swap and footrule distances also follow from the analogous reductions from \TSAT.
\end{proof}
}

%% file: conclusion.tex
\section{Conclusion and Future Direction}\label{sec:con}

In this paper, we have proposed a new model of bribery. We have argued that the bribery models studied so far in computational social choice may fail to capture intricacies in certain situations, for example, where there is a fear of information leakage and voters care about social reputation. We have discussed how our \ODB and \OB problems can suitably model those scenarios. We have then shown that the \ODB problem is polynomial time solvable for the plurality and veto voting rules for the swap, footrule, and maximum displacement distances, and for the $k$-approval voting rule for the swap distance if the distance allowed is $1$ (and thus for the footrule distance, it is $3$). For the $k$-approval and simplified Bucklin voting rules for the maximum displacement distance, we have shown that the \OB problem is polynomial time solvable. We have then proved that the \OB problem (and thus the \ODB problem) is \NPC for the $k$-approval and simplified Bucklin voting rules for the swap distance even if the distance allowed is $2$ (and thus for the footrule distance, it is $4$), for a class of scoring rules which includes the Borda voting rule, maximin, Copeland$^\alpha$ for any $\alpha\in[0,1]$, and Bucklin voting rules for the swap and maximum displacement distances even when the distance allowed is $1$ (and thus for the footrule distance, it is $3$). In particular, we have proved tight (in terms of \delta) computational complexity results for the \OB problem for the swap, footrule, and maximum displacement distances. Our results show that the notion of optimality makes bribery much richer than optimal manipulation in the sense that the complexity of the \OB problem for some commonly used voting rule ($k$-approval for example) can change drastically if we change the measure of distance under consideration.

\longversion{It would be interesting to find approximation algorithms for the \OB problem where it is \NPC (our hardness proofs already show APX-hardness). The bribery problem in this paper can be extended by introducing a pricing model and a (global) budget for the briber\longversion{ which the briber needs to respect}. In any setting where our \OB problem is \NPC, hardness in such sophisticated models in the corresponding setting will immediately follow. However, it would be interesting to extend our polynomial time algorithms to those\longversion{ sophisticated} models.}